\documentclass[journal,12pt,onecolumn,draftclsnofoot,]{IEEEtran}
\usepackage[T1]{fontenc}
\usepackage{float}
\usepackage{amsmath}
\usepackage{amsthm}
\usepackage{amssymb}
\usepackage{stackrel}
\usepackage{graphicx}
\usepackage{setspace}
\usepackage{url}
\usepackage{bm}
\usepackage{caption}
\usepackage{subfigure}
\usepackage{cases}
\usepackage{xcolor}
\usepackage{diagbox}

\usepackage{algorithm}
\usepackage{algorithmic}
\usepackage{stfloats}
\allowdisplaybreaks[4]

\makeatletter

\floatstyle{ruled}
\newfloat{algorithm}{tbp}{loa}
\providecommand{\algorithmname}{Algorithm}
\floatname{algorithm}{\protect\algorithmname}

\theoremstyle{plain}
\newtheorem{theorem}{Theorem}

\theoremstyle{definition}

\theoremstyle{plain}

\theoremstyle{plain}

\IEEEoverridecommandlockouts
\usepackage{ifpdf}
\usepackage{cite}
\ifCLASSOPTIONcompsoc
\usepackage[caption=false,font=normalsize,labelfont=sf,textfont=sf]{subfig}
\else
\usepackage[caption=false,font=footnotesize]{subfig}
\fi
\usepackage{amsmath}
\usepackage{mathtools}
\usepackage{booktabs}
\usepackage{multirow}
\usepackage{setspace}
\usepackage{lettrine}
 \usepackage{array}
\@ifundefined{showcaptionsetup}{}{
	\PassOptionsToPackage{caption=false}{subfig}}
\usepackage{subfig}
\makeatother
\usepackage{babel}

\newtheorem{remark}{Remark}
\newtheorem{proposition}{Proposition}

\begin{document}
	\captionsetup[figure]{font={small}, name={Fig.}, labelsep=period}
	
	\title{Multi-Carrier NOMA-Empowered Wireless Federated Learning with Optimal Power and Bandwidth Allocation}
	\author{Weicai~Li,~Tiejun~Lv,~Yashuai~Cao,~Wei~Ni,~and~Mugen~Peng,~\emph{Fellow,~IEEE}
	
\thanks{W. Li and T. Lv are with the School of Information and Communication Engineering, Beijing University of Posts and Telecommunications (BUPT), Beijing 100876, China (e-mail: \{liweicai, lvtiejun\}@bupt.edu.cn).}
\thanks{Y. Cao is with the Department of Electronic and Communication Engineering, North China Electric Power University (NCEPU), Baoding 071003, Hebei, China (e-mail: yashcao@ncepu.edu.cn).}
\thanks{W.~Ni is with Data61, Commonwealth Science and Industrial Research Organisation (CSIRO), Sydney, New South Wales, 2122, Australia (e-mail: wei.ni@data61.csiro.au).}	
\thanks{M.~Peng is with State Key Laboratory of Networking and Switching Technology, Beijing University of Posts and Telecommunications, Beijing 100876, China (e-mail: pmg@bupt.edu.cn).}
	}
	
	\maketitle
	\begin{abstract}
Wireless federated learning (WFL) undergoes a communication bottleneck in uplink, limiting the number of users that can upload their local models in each global aggregation round. 
This paper presents a new multi-carrier non-orthogonal multiple-access (MC-NOMA)-empowered WFL system under an adaptive learning setting of Flexible Aggregation. 
Since a WFL round accommodates both local model training and uploading for each user, the use of Flexible Aggregation allows the users to train different numbers of iterations per round, adapting to their channel conditions and computing resources.
The key idea is to use MC-NOMA to concurrently upload the local models of the users, thereby extending the local model training times of the users and increasing participating users. 
A new metric, namely, Weighted Global Proportion of Trained Mini-batches (WGPTM), is analytically established to measure the convergence of the new system. 
Another important aspect is that we maximize the WGPTM to harness the convergence of the new system by jointly optimizing the transmit powers and subchannel bandwidths.
This nonconvex problem is converted equivalently to a tractable convex problem and solved efficiently using variable substitution and Cauchy's inequality.
As corroborated experimentally using a convolutional neural network and an 18-layer residential network, the proposed MC-NOMA WFL can efficiently reduce communication delay, increase local model training times, and accelerate the convergence by over 40\%, compared to its existing alternative.

	\end{abstract}
	
{}

	\section{Introduction}
	
\lettrine[lines=2]{B}{eing} a new and promising distributed machine learning (ML) framework, federated learning (FL) is able to protect user privacy, alleviate computing pressure, and reduce response delay by training models in a decentralized manner without sharing raw private data~\cite{wahab_federated_2021}. 
FL has been increasingly applied to wireless networks, referred to as wireless FL (WFL), where multiple wireless users train an ML model collaboratively, e.g., for intelligent transportation~\cite{manias2021making,Li2016Energy}, smart surveillance~\cite{9063670, Li2019Energy}, and many other Internet-of-Things (IoT) applications~\cite{9460016,Wang2015VANET,Li2019On}.
With the assistance of a central server, the wireless users train the model round by round. In each WFL round, the users start with the global model provided by the central server and update the model locally based on their local datasets. At the end of the round, the users upload their locally updated models to the server, where the local models are aggregated to update the global model.

A critical challenge arising is a communication bottleneck in the uplink of a WFL system, resulting from relatively limited system bandwidth~\cite{lim_federated_2020}. Specifically, a ``large-is-better" conclusion was drawn in~\cite{yang_training_2020,zeng_energy-efficient_2020}; i.e., a larger number of participating users or a larger size of the dataset leads to a higher training efficiency of WFL.  
 In~\cite{xu_client_2021},  a ``later-is-better” phenomenon was observed that fewer clients participating in the early rounds of WFL and more in the later rounds can help WFL achieve a better accuracy, lower training loss, and better robustness.
 Unfortunately, typical orthogonal multiple-access (OMA) limits the number of users that can upload their local models for global model aggregation per WFL round~\cite{7973146}.

Another critical challenge is the lack of joint design of model training and uploading in WFL. 
Particularly, wireless channels can differ substantially among users and change over time; i.e., the users can incur different delays in uploading local models and consequently, their local training times differ in a WFL round~\cite{wahab_federated_2021}. 
The conventional synchronous FL (Sync-FL) required all users to complete their local training before uploading their local models synchronously for global aggregation~\cite{abutuleb_joint_2020,pmlr-v54-mcmahan17a,9264742}. This is rigid and does not suit WFL.
Recently, asynchronous FL (Async-FL) was developed to allow users to upload their local models asynchronously after completing their local training~\cite{fedasync,Fed-AT,TT-fed,9725259}. 
However, there could still be non-negligible gaps between the time the local models are uploaded and the time the updated global model is announced, leading to the under-use of the computing powers of the users. 
An alternative to Async-FL is incomplete aggregation~\cite{pmlr-v130-ruan21a}, where all users upload their local models synchronously, and some of the models are trained incompletely with fewer iterations than others. 
Unfortunately, no consideration has been given to the model uploading under incomplete aggregation.

 Existing designs of WFL have focused primarily on three key performance indicators (KPIs), namely, model accuracy, convergence speed, and energy efficiency, separately in most cases, as summarized in Table \ref{tab:kpi}. 
Attempts have been witnessed to improve the training accuracy of WFL, typically  by increasing the number of participating users~\cite{zeng_energy-efficient_2020,xu_client_2021} and/or the size of the datasets involved in training~\cite{yang_training_2020}.
The authors of~\cite{yang_training_2020} and~\cite{zeng_energy-efficient_2020} revealed that large datasets or a large number of participating users contribute to the training efficiency of WFL, respectively; i.e., ``larger-is-better."
In~\cite{xu_client_2021},  the impacts of the ``later-is-better” phenomenon, i.e., involving more training data at later stages of WFL, were analyzed quantitatively on the training accuracy, loss, and robustness of WFL.

 Some studies have expedited the convergence of WFL by maximizing the communication efficiency, e.g., data rate or throughout. This is due to the fact that shorter communication time helps either reduce the duration of an FL round, or extend the local model training time of the users within a round. In~\cite{9322270}, the weighted communication rate of all participating users was maximized for the fast convergence of WFL by formulating a maximum-weight independent set problem that was solved approximately based on graph theory.

In \cite{9718086}, a base station (BS) was designed to transfer wireless power to energize edge devices for local model training and uploading. 
The schedule of the wireless power transfer was optimized to minimize a system-wise cost, accounting for both the energy consumption and FL convergence latency. 
 Async-FL was also designed to allow each individual user to upload its local model whenever completing its local training~\cite{fedasync,Fed-AT,TT-fed,9725259}. This helps reduce latency and thus speeds up convergence.

\begin{table}[]\small
\centering
\caption{The summary of related works from the perspectives of design objectives.}
\label{tab:kpi}
\begin{tabular}{c|ccccc|c|c}
\hline
\multirow{2}{*}{\diagbox[innerwidth=3.3cm]{KPI}{Reference}} & \multicolumn{5}{c|}{Sync. FL} & \begin{tabular}[c]{@{}c@{}}Async. \\ FL\end{tabular} & \begin{tabular}[c]{@{}c@{}}Flexible \\ Aggregation\end{tabular} \\ \cline{2-8} 
 & \multicolumn{1}{c|}{\cite{yang_training_2020,8952884}} & \multicolumn{1}{c|}{\cite{zeng_energy-efficient_2020,xu_client_2021}} & \multicolumn{1}{c|}{\cite{9264742,9764370,9844152}} & \multicolumn{1}{c|}{\cite{9322270}} & \multicolumn{1}{c|}{\cite{9718086}}  & \cite{fedasync,Fed-AT,TT-fed,9725259}& This paper \\ \hline
Accuracy & \multicolumn{1}{c|}{$\checkmark$} & \multicolumn{1}{c|}{$\checkmark$} & \multicolumn{1}{c|}{} & \multicolumn{1}{c|}{} & \multicolumn{1}{c|}{} &  & $\checkmark$ \\ \hline
Convergence speed & \multicolumn{1}{c|}{} & \multicolumn{1}{c|}{}  & \multicolumn{1}{c|}{} & \multicolumn{1}{c|}{$\checkmark$} & \multicolumn{1}{c|}{$\checkmark$} &  $\checkmark$ & $\checkmark$ \\ \hline
Energy efficiency & \multicolumn{1}{c|}{} & \multicolumn{1}{c|}{$\checkmark$} & \multicolumn{1}{c|}{$\checkmark$} & \multicolumn{1}{c|}{} & \multicolumn{1}{c|}{$\checkmark$} &  & $\checkmark$ \\ \hline
\end{tabular}
\end{table}

Other studies have been devoted to reducing the energy consumption of WFL, including both computing and communication energy consumption.
In~\cite{9264742}, the total energy consumption of all users was reduced under a latency constraint by developing a suboptimal, low-complexity, iterative algorithm.
In~\cite{9764370}, deep reinforcement learning was employed to minimize the cost of mobile edge computing-assisted hierarchical FL by resource allocation and IoT device orchestration. 
In~\cite{9844152}, the selection of IoT devices and relays, and the transmit powers and CPU frequencies of the selected IoT devices were optimized to minimize the energy consumption, subject to the delay constraint of the FL. A graph-theoretic approach was taken to design a low-complexity, suboptimal solution by applying a greedy maximum-weight-independent-set algorithm.

On the other hand, non-orthogonal multiple-access (NOMA), especially multi-carrier NOMA (MC-NOMA), is an advanced multiple-access technique for multiplexing multiple users in the power domain in every subchannel of a multi-channel system~\cite{7557079}.
It can admit more users to transmit concurrently, thereby increasing the volume of training data involved in model training, compared to the state-of-the-art OMA systems, such as (orthogonal) frequency-division multiple access (OFDMA/FDMA), time-division multiple-access (TDMA), and code-division multiple-access (CDMA). 
Moreover, the concurrent transmissions of the users under MC-NOMA also prevent the awkward situation (i.e., under TDMA) that the users scheduled to upload their local models earlier have to shorten their local training time.  
Despite NOMA being considered for local model transmissions in~\cite{9718086,9764370,9844152}, it was not jointly considered with the learning parameter selection and their impact on the convergence of WFL was overlooked.

With different processing speeds, users may have to upload their local models asynchronously. 
Async-FL~\cite{fedasync} allowed every user to upload its local model whenever completing its local training.
The FL server aggregated the local model and the latest global model at a cost of stability and convergence delay compared to Sync-FL. 
A few improvements of Async-FL are Federated Learning with Asynchronous Tiers (Fed-AT)~\cite{Fed-AT} and Time-Triggered Federated Learning (TT-Fed)~\cite{TT-fed}. 
Fed-AT and TT-Fed divided an FL process into multiple Sync-FL parallel processes among users with similar processing speeds. In the case of Fed-AT, the global models of the multiple parallel Sync-FL processes were aggregated asynchronously~\cite{Fed-AT}. 
In the case of TT-Fed, 
the global models were synchronously aggregated among some parallel Sync-FL processes that were due to aggregate in the same round~\cite{TT-fed}. In~\cite{9725259}, retransmissions of local models were enabled in case of transmission collisions. The weighting coefficients of the local models were optimized to capture the freshness of the models.
In~\cite{fedasync,Fed-AT,TT-fed,9725259}, the channel between a user and the FL server was assumed to be time-invariant, and so was the local model training time of the user. Moreover, the users did not resume local training until scheduled global aggregation and announcement in~\cite{Fed-AT,TT-fed,9725259}, leading to inefficient use of computing resources. 

Different from Sync-FL and Async-FL, Flexible Aggregation~\cite{pmlr-v130-ruan21a} maintained a consistent duration of the global aggregation rounds like Sync-FL and extended FedAVG by allowing flexible epochs in each global aggregation round. The users could run different numbers of iterations before uploading their local models. To this end, Flexible Aggregation offers flexible local training settings and efficient use of the users' computing powers, compared to Async-FL~\cite{fedasync,Fed-AT,TT-fed,9725259}. However, the transmission delays of local model uploading were overlooked in~\cite{pmlr-v130-ruan21a}, which would inevitably impact the local training time within a round.

In a different yet relevant context, over-the-air (OTA) computation has been integrated into WFL systems~\cite{8952884,9484466,9829190,ota}. OTA-FL exploits the superposition property of wireless multiple-access channels and aggregates local models in the analog domain. 
Despite the use of NOMA-based superposition of the local models, the MC-NONA WFL considered in this paper is distinctively different from OTA-FL in the sense that the local models need to be recovered individually and aggregated in the digital domain (as opposed to the analog domain in OTA-FL). 
The digital operation is known to be superior in practicality and hardware cost~\cite{SHI}.

This paper presents a new WFL paradigm, which incorporates MC-NOMA into a WFL system to holistically consider all the three KPIs of model accuracy, converge speed, and energy efficiency. 
The system supports Flexible Aggregation to allow users to train different numbers of iterations in each WFL round, adapting to the channel conditions and computing powers of the users. 
As a result, the communication bottleneck of WFL systems can be substantially alleviated, and the computing powers of the users can be efficiently utilized. 

The key contributions of the paper are listed as follows. 
\begin{itemize}
    \item 
    We put forth a new MC-NOMA-empowered WFL system with Flexible Aggregation, where multiple users upload their local models concurrently in each subchannel, thereby extending the local model training times of the users and increasing participating users.

    \item 
     A new metric, namely, Weighted Global Proportion of Trained Mini-batches (WGPTM), is defined to measure the convergence of the new system based on the analysis of the convergence upper bound.

    \item 
    A new problem is formulated to maximize the WGPTM and harness the convergence of the MC-NOMA WFL by optimizing the transmit powers of the users and the bandwidths of the subchannels.
    The problem is non-trivial due to its nonconvexity. 

    \item 
    By employing variable substitution and Cauchy's inequality (in particular, the equality condition of Cauchy's inequality), we establish the necessary conditions of the optimal solution to the nonconvex problem and then transform the problem consequently into a convex problem solved with polynomial complexity.

\end{itemize}

 The proposed approach is comprehensively tested using a convolutional neural network (CNN) and an 18-layer residential network (ResNet18) on the Federated Extended MNIST (F-MNIST) and Federated CIFAR100 (F-CIFAR100) datasets. 
It is shown that under the Flexible Aggregation, the new MC-NOMA WFL with the optimal joint power and bandwidth allocation speeds up convergence by 10\% and 22\%, compared to MC-NOMA WFL with the optimal power allocation only and an MC-OMA WFL, respectively. Moreover, the proposed MC-NOMA WFL with Flexible Aggregation can outperform its potential existing alternatives, i.e., MC-OMA WFL under Sync-FL and Async-FL, by over 60\% and 40\%, respectively.

The rest of this paper is organized as follows. 
In Section~\ref{section:system}, the system model is described, the convergence upper bound is analyzed, and the new metric is defined. In Section \ref{section:problem}, the new problem is formulated to jointly optimize the power and bandwidth allocations and its solution is derived. Simulation results are presented in Section~\ref{section: simulation}, followed by conclusions in Section~\ref{section:con}. The notations used are collated in Table \ref{notations}.

	\begin{figure}[t]
		\centering{}
		\includegraphics[scale=0.4]{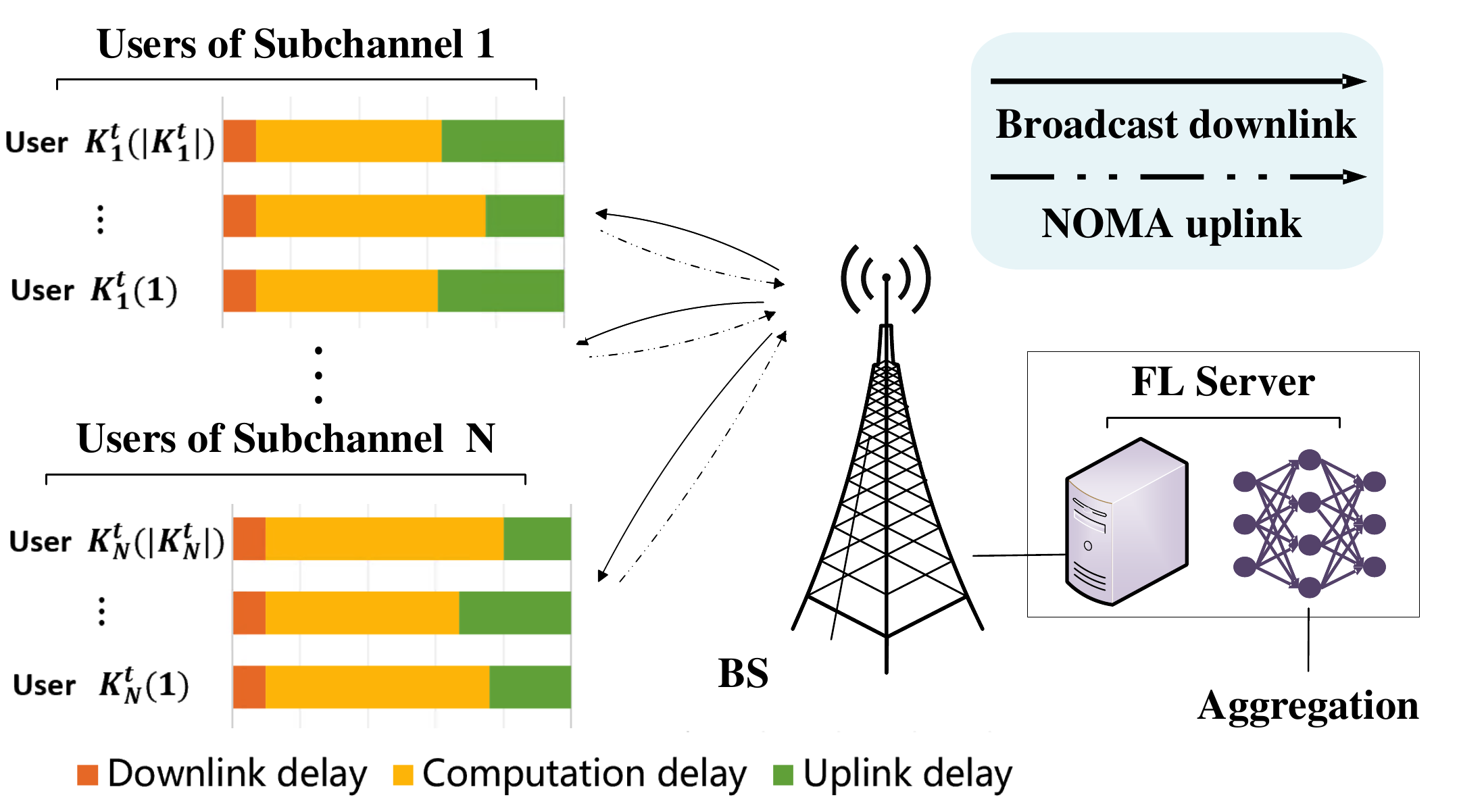}
		\caption{An illustration of the new MC-NOMA WFL system,  
  where all users have the same downlink delay and total delay, but different users have different uplink and computation delays (on the left of the figure).}
		\label{fig:system}
	\end{figure}
	
	\section{System Model}\label{section:system}
	The considered MC-NOMA WFL system consists of $K$ edge users and an FL server. The edge users train their local models based on their local data. The FL server aggregates the local models uploaded by the edge users to produce the global model and then broadcasts the global model to all users. 
 The FL server is hard-wired to a BS. The edge users upload their local models to and receive the global model from the FL server via wireless channels between the users and BS,
using MC-NOMA in the uplink and broadcast in the downlink, respectively.

	Let $\mathcal{K}$ collect the indexes to all users, $\mathcal{D}_k$ be the local dataset of user $k\in\mathcal{K}$, and $\mathcal{D}=\underset{{k\ensuremath{\in\mathcal{K}}}}{\sum}\mathcal{D}_{k}$ be the global dataset that collects the local dataset $\mathcal{D}_k$ of all users.
	Each user $k\in {\cal K}$ can define its objective function $F_k(\mathbf{w})$ that is the average loss function over its local dataset ${\cal D}_k$. Here, $\mathbf{w}\in \mathbb{R}^{d}$ is the model parameter vector, and $d$ is the dimension of the FL model parameter. The minimum value of $F_k(\mathbf{w})$ is denoted by $F_k^*$. The global objective of the FL is to minimize a weighted loss function of all users, i.e., $\underset{\mathbf{w}}{\min} \left(F(\mathbf{w})\right)=\underset{\mathbf{w}}{\min} \underset{k\in\mathcal{K}}{\sum}e_{k}F_k(\mathbf{w}),$ where $e_k=\frac{|{\cal D}_{k}|}{|\mathcal{D}|}$ is the weighting coefficient of user $k$, $|\cdot|$ denotes the cardinality of a set, and $|\mathcal{D}|=\underset{{k\ensuremath{\in\mathcal{K}}}}{\sum}|\mathcal{D}_k|$.

	 \begin{table}[t]

	\centering
	\caption{Notation and Definition}\label{notations}
	\begin{tabular}{ m{1.5cm}<{\centering}|m{14cm}  }
		\hline
		\textbf{Notation}& \multicolumn{1}{c}{\textbf{Definition}} \\
		\hline
		$\mathcal{K}$, $\mathcal{K}_{n}$&The set of all users, and the set of users assigned to the $n$-th subchannel; \\\hline
		$\mathcal{D}$, $\mathcal{D}_k$&Global dataset, and the local dataset of user $k$; \\\hline
		$\mathcal{M}$, $\mathcal{M}_k$&Global mini-batch set, and the local mini-batch set of user $k$; \\\hline
		$e_k$& Weighting coefficient of user $k$; \\\hline
		$F$, $F_k$&Global FL objective function and the FL objective function of user $k$; \\\hline
		$\eta$, $\eta_k$&The learning rate of all users for training an iteration and the learning rate of user $k$ for training a mini-batch;\\\hline
		$\mathbf{w}$&FL model parameter vector;\\\hline
		$\phi_{k}^t$ & Number of trained mini-batches by user $k$ in the $t$-th WFL round;\\\hline
		$\Phi_{k}^t$, $\Phi_{t}$ & Local Proportion of Trained Mini-batches (LPTM) of user $k$ and Weighted Global Proportion of Trained Mini-batches (WGPTM) in the $t$-th WFL round; \\\hline
		${T_{k}}$, $T_{{\rm{u}},k}$, $T_{{\rm c},k}$, $T_{{\rm{d}}}$&The total delay, uplink delay, computation delay, and downlink delay of user $k$ in a WFL round;\\\hline
		$h_{n,i}$& The channel gain of user $i$ on the $n$-th subchannel in the current WFL round; \\\hline
		$B$, $B_n$&Total uplink bandwidth and the bandwidth of the $n$-th subchannel;  \\\hline
		$p_{n,i}$, $\mathbf{p}_{n}$& The transmit power of the $i$-th user and the set of transmit powers on the $n$-th subchannel in a   round; \\\hline
		$P_{\max }$& The maximum transmit power of a user. \\\hline
		\end{tabular}
	\end{table}
	\subsection{WFL Workflow and Delay Model}
	
	The training process of the MC-NOMA WFL system is divided into rounds.
	 A WFL round must accommodate both the local model training and uploading of each user. When the channel condition is poorer between a user and the server, the transmit rate of the user is lower and the user spends longer time within a WFL round in uploading its local model. As a consequence, shorter time can be used to train the local model within the round.

	\subsubsection{Communication Protocol and Delay}
	
Assume that the channels experience block fading; in other words, the channel coefficients remain unchanged during a WFL round, and change independently between WFL rounds.
	At the beginning of the $t$-th WFL round, the FL server broadcasts the latest global model $\mathbf{w}_{\mathcal{G}}^{t-1}$ to all users. The downlink delay, denoted by $T_{{\rm{d}}}^t$, is consistent across all users.	
	Upon the receipt of the global model, each user $k$, $\forall k\in\mathcal{K}$ trains its local model $F_k(\mathbf{w})$ based on the global model $\mathbf{w}_{\mathcal{G}}^{t-1}$ and its local data ${\cal D}_k$. The user updates and uploads its local model $\mathbf{w}_k^t$ to the FL server by the end of the $t$-th WFL round.
	The FL server aggregates the local models of the users and updates the global model, i.e., 
		$\mathbf {w}_{\mathcal{G}}^{t}=\sum\limits _{k\in\mathcal{K}}e_{k}\mathbf {w}_{k}^{t}$.
	The global model ${\bf{w}}_{\mathcal{G}}^t$ is broadcast to the users at the beginning of the $(t+1)$-th WFL round.

    Suppose that the system bandwidth is $B$~Hz in the uplink during the $t$-th WFL round. The bandwidth is divided into $N$ subchannels with configurable bandwidths. 
	The index to the subchannels is $n =1,\cdots,N$. 
	Each user is assigned to a subchannel according to a NOMA user clustering algorithm, e.g., the one developed in~\cite{7557079}.
	Power-domain NOMA (PD-NOMA) is performed in each subchannel.	
	The bandwidth of the $n$-th subchannel is $B^t_{n}$ Hz
	in the $t$-th WFL round. 
	Let $\mathcal{K}^t_{n}$ collect the indexes to the users assigned to the $n$-th subchannel in the round.  $|\mathcal{K}^t_{n}|$ is the number of users assigned to the $n$-th subchannel. $\mathcal{K}=\cup_{n = 1}^N\mathcal{K}^t_{n}$. $K=\sum_{n = 1}^N|\mathcal{K}^t_{n}|$.
	For illustration convenience, the users are sorted in the ascending order of their channel gains in $\mathcal{K}^t_{n}$, i.e., $|h^t_{n,1}|^{2}\le\cdots\le|h^t_{n,|\mathcal{K}^t_{n}|}|^{2}$, where $h^t_{n,i}$ is the channel coefficient from the $i$-th user in $\mathcal{K}^t_{n}$, i.e., user $\mathcal{K}^t_{n}(i)$, to the BS on the $n$-th subchannel in the $t$-th WFL round.

	At the BS (to which the FL server is attached), successive interference cancellation (SIC) \cite{7959539} is carried out to decode the uploaded local model parameters of the users. The users assigned to a subchannel are decoded and then canceled in the descending order of their channel gains. The received signal of user $\mathcal{K}^t_{n}(i)$ at the BS is given by 
	\begin{align}
		y^t_{n,i}=  &h^t_{n,i}\sqrt {p^t_{n,i}} s^t_{n,i} + \sum\limits_{j=1}^{i-1} h^t_{n,j}\sqrt {p^t_{n,j} } s^t_{n,j}+ z^t_{n}, 
		\label{eq:y_u}
	\end{align}	
	where $p^t_{n,i}$ is the transmit power of the $i$-th user in $\mathcal{K}^t_{n}$, i.e., user $\mathcal{K}^t_{n}(i)$, in the $t$-th WFL round; $s^t_{n,i}$ is the encoded signal of the user's local model parameter ${\mathbf{w}}^t_{\mathcal{K}^t_{n}(i)}$; and $z^t_{n}$ is the band-limited additive white Gaussian noise in the $n$-th subchannel. The noise power spectral density is $N_0$. The noise power is ${N_0}{B^t_{n}}$. 
	For ease of exposition, let $\mathbf{p}^t_{n} = \{ p^t_{n,i},\,\forall {i}={1,\cdots,|\mathcal{K}^t_{n}|}\}$ collect the transmit powers of all users in $\mathcal{K}^t_{n}$.
	With the SIC, the SINR of user $\mathcal{K}^t_{n}(i)$ is written as
	\begin{align}
		\Upsilon ^t_{n,i} = {{{\left| {h^t_{n,i}} \right|}^2}p^t_{n,i}}\Big/\Big({\sum\limits_{j=1}^{i-1} {{{\left| {h^t_{n,j}} \right|}^2}p^t_{n,j} + {N_0}{B^t_{n}}} }\Big). 
		\label{eq:YY_u}
	\end{align}
	The uplink transmit rate of user $\mathcal{K}^t_{n}(i)$ is given by
	\begin{align}
		R^t_{n,i}  = {B^t_{n}}{\log _2}\left( {1 + \Upsilon ^t_{n,i} } \right),\,i=1, \cdots, |\mathcal{K}^t_{n}|. \label{eq:R_u}
	\end{align}
	The corresponding local model uploading delay is 
	\begin{equation}\label{eq:T_u}
		 T_{{\rm u},{\mathcal{K}^t_{n}(i)}}^t = \frac{{S}}{R^t_{n,i}}	={\frac{{S}}{{{B^t_{n}}{\log }_2}\Big( 1 + \Upsilon ^t_{n,i}  \Big)}}=S\Big/{{{B^t_{n}}{\log }_2}\bigg( {\frac{{\sum\limits_{j=1}^i {{{\left| {h^t_{n,j}} \right|}^2}p^t_{n,j} + {N_0}{B^t_{n}}} }}{{\sum\limits_{j = 1}^{i - 1} {{{\left| {h^t_{n,j}} \right|}^2}p^t_{n,j} + {N_0}{B^t_{n}}} }}} \bigg)},
	\end{equation}
	where $i=1, \cdots, |\mathcal{K}^t_{n}|$; and $S$ is the size of the encoded local model parameters $\{s^t_{n,i},\forall i=1,\cdots,|\mathcal{K}_n^t|,\forall n=1,\cdots,N\}$, which is consistent across all users. 

	\subsubsection{Computation Model}
    Consider Flexible Aggregation~\cite{pmlr-v130-ruan21a}, the new MC-NOMA WFL system allows the users to execute different numbers of iterations within a round. (The local model of a user is trained on its local dataset once per iteration.) A user can even upload its latest local model parameter if the model parameter is obtained in the middle of the last iteration. 
      With this flexibility, the user can decide its local model training time and model uploading time adapting to the changes in its  channel conditions. This allows the user to fully utilize their computing and communication resources to substantially increase the number of users involved and the amount of training data utilized in WFL.
     
    To implement Flexible Aggregation, each user $k$ divides its dataset ${\cal D} _k$ evenly into independent and identically distributed (i.i.d.) mini-batches. Each mini-batch comprises $M$ samples and obeys the same distribution as $\mathcal{D}_k$.  Let $\mathcal{M}_k$ collect all mini-batches of training data at user $k$. Then, $|{\cal M}_k|$ is the number of mini-batches in $\mathcal{M}_k$.
	$|{\cal D} _k|=|{\cal M}_k|\times M$.	
 Let $\eta$ be the learning rate of an iteration, which is consistent across all users. $\eta_k=\eta/| {\cal M}_k|$ is the learning rate at which user $k\in\mathcal{K}$ trains a mini-batch.
	The number of complete iterations and the number of mini-batches trained in the last, incomplete iteration can be configured for each user to accommodate its local training and model uploading within a round, based on the channel conditions of the user.	

	The number of floating-point operations (FLOPs) is typically used to measure computational complexity. 
	The number of floating-point operations per second (FLOPS) measures the processing speed of a user.
	Let $\alpha$ be the FLOPs required to train a local model based on a mini-batch of local data, and $\phi_{k}^t$ be the number of mini-batches trained at user $k \in \mathcal{K}$ in the $t$-th WFL round. ${\alpha_{k}^t} =\phi_{k}^t {\alpha}$ is the FLOPs needed at the user in the round. The local training delay of user $k$ is
	\begin{align}
		T_{{\rm{c}},k}^t = \alpha_{k}^t/\beta_k = \phi_{k}^t\alpha/ \beta_k ,\label{eq:fl delay}
	\end{align}
	where $\beta_k$ (in FLOPS)  is the processing speed of user $k\in \mathcal{K}$.

	\subsubsection{Total Delay}
	In the $t$-th WFL round, the total delay of user $k \in \mathcal{K}$ is defined as
	\begin{align}
		{T_{k}^t} = T_{{\rm{u}},k}^t + T_{{\rm c},k}^t+ T_{{\rm{d}}}^t.
		\label{eq:Time}
	\end{align}
	The users have the same downlink delay, and different computation and uplink delays, as illustrated in Fig. \ref{fig:system}. A hard requirement of WFL is that all users complete uploading their local models by the end of a WFL round. Nevertheless, the users assigned to a subchannel can have different starting times for model uploading\footnote{Some users start transmitting their local models earlier than the others in a subchannel due to their relatively poorer channels. These users would undergo weaker co-channel interference and enjoy higher transmit rates than \eqref{eq:R_u} before the others start. The transmission schedules produced by the algorithm proposed in this paper remain valid.}.

	\subsection{Convergence Analysis}\label{convergence}
	In a WFL round, we define the Local Proportion of Trained Mini-batches (LPTM)
 of a user to be the ratio of the number of locally trained mini-batches in the round to the total number of mini-batches in the local dataset of the user. 
	In the $t$-th WFL round, the LPTM of user $k$ is
	\begin{align}
		\Phi_{k}^t=\phi_{k}^t\Big/|{\cal M}_{k}|=\phi_{k}^t\Big/e_{k}\underset{k\in\mathcal{K}}{\sum}|\mathcal{M}_{k}|,\label{iteration}
	\end{align}
	where $\phi_{k}^t$ is the number of locally trained mini-batches at user $k$ in the $t$-th WFL round. The second equation in \eqref{iteration} is because $|{\cal D} _k|=|{\cal M}_k|\times M$ and hence $e_k=|{\cal D} _k|/|{\cal D}|=|{\cal M} _k|/\sum_{k\in {\cal K}}|{\cal M}_k|$.

Assume that there exists such a constant $m>0$ that local objective functions $F_k,\, k\in\mathcal{K}$ and the global objective function $F$ are $m$-strongly convex for all $x\in\mathbb{R}^{d}$, i.e., 
	\begin{align}
		\nabla^{2}f\left(x\right)\succeq mI,\forall f \in \{F_k,\, k\in\mathcal{K}\}\cup \{F\},\label{ass1}
	\end{align}
	where $\nabla$ stands for gradient and $I$ is the identity matrix.
    Also assume that there is such a constant $L>0$ that, if $F_k, \, k\in\mathcal{K}$, and $F$ are $L$-smooth, for any $x\in\mathbb{R}^{d}$, we have 
	\begin{align}
		\nabla^{2}f\left(x\right)\preceq LI,\forall f \in \{F_k,\, k\in\mathcal{K}\}\cup  \{F\}.\label{ass2}
	\end{align}
	Further, assume that there exists such a constant $G$ that the expected squared norm of stochastic gradients is uniformly bounded: 
	\begin{equation}
		\mathbb{E}\{\left\Vert \nabla F_{k}\right\Vert ^{2}\}\leq G^{2},\,\forall k\in{\cal K},    \label{ass3}
	\end{equation}
	where $\mathbb{E}\{\cdot\}$ takes expectation and $\Vert\cdot \Vert$ denotes $\ell_1$-norm.
 
Note that these assumptions have been widely considered in the literature for the convergence analysis of FL, e.g.,~\cite{9264742,pmlr-v130-ruan21a,boyd2004convex}. By joining Assumptions~\eqref{ass1} and~\eqref{ass2}, i.e., $\nabla^{2}f\left(x\right)$ is upper and lower bounded, the objective function is strongly convex and smooth, based on which we can analyze the convergence of FL and derive the optimality gap. On the other hand, strongly convex and smooth objective functions have been widely applied in ML models, e.g., multi-layer perceptron (MLP) and support vector machine (SVM).

	Under the assumptions, \eqref{ass1}--\eqref{ass3}, the relation between the convergence bound at the  $t$-th round of the proposed WFL system, and the LPTM $\Phi_{k}^t,\,\forall k \in\mathcal{K}$, is established in the following theorem.

	\begin{theorem}
		After local training and global aggregation in the $t$-th WFL round, the global model parameter is ${\bf{w}}_{\mathcal{G}}^t$. With the global loss function $F$ minimized in the $t$-th WFL round, the upper bound of the difference between $\mathbb{E}(F(\mathbf{w}_{\mathcal{G}}^{t}))$ and $F\left(\mathbf{w}_{\mathcal{G}}^{t-1}\right)$ satisfies
		\begin{align}
	\mathbb{E}\left(F\left(\mathbf{w}_{\mathcal{G}}^{t}\right)\right) 	-F\left(\mathbf{w}_{\mathcal{G}}^{t-1}\right)  
			\leq&-\Big(\underset{k\in{\cal K}^t_{n}}{\min}\{c_k\mathbb{E}\left(F_{k}\left(\mathbf{w}_{\mathcal{G}}^{t-1}\right)-F_{k}^{*}\right)\}-\eta G^{2}\Big)\underset{k\in\mathcal{K}}{\sum}e_{k}\Phi_{k}^t	 \notag \\
				&+\underset{k\in\mathcal{K}}{\sum}e_{k}\mathcal{O}\left((c_k)^{2}\right)\mathbb{E}\left(F_{k}\left(\mathbf{w}_{\mathcal{G}}^{t-1}\right)-F_{k}^{*}\right).
			\label{theorem}
		\end{align}
  Here, $c_k= 2m\eta_k - mL\eta_k^2$. The upper bound in \eqref{theorem} is proportional to $-\underset{k\in\mathcal{K}}{\sum}e_{k}\Phi_{k}^t$, i.e., 
		\begin{align}
			\Big(\underset{k\in{\cal K}^t_{n}}{\min}\{c_k\mathbb{E}\left(F_{k}\left(\mathbf{w}_{\mathcal{G}}^{t-1}\right)-F_{k}^{*}\right)\}&-\eta G^{2} \Big)\underset{k\in\mathcal{K}}{\sum}e_{k}\Phi_{k}^t	\propto \underset{k\in\mathcal{K}}{\sum}e_{k}\Phi_{k}^t.
		\end{align}		
	\end{theorem}
	\begin{proof}\let\qed\relax
		See Appendix~\ref{proof-theorem}. 	$\hfill\blacksquare$
	\end{proof}

	\begin{remark}
		In Theorem 1, $\mathbb{E}\left\{F\left(\mathbf{w}_{\mathcal{G}}^{t}\right)\right\}	-F\left(\mathbf{w}_{\mathcal{G}}^{t-1}\right)$ measures the decrease of the global loss function $F(\mathbf{w})$ in the $t$-th round. The larger $\underset{k\in\mathcal{K}}{\sum}e_{k}\Phi_{k}^t$ is, the faster the global loss function decreases. To this end, $\sum\limits _{k\in\mathcal{K}}e_{k}\Phi_{k}^t$ can be used to measure the training effect within a round. 
	\end{remark}

	In light of \textbf{Remark 1}, we define the new metric, termed Weighted Global Proportion of Trained Mini-batches (WGPTM), to measure the convergence of WFL depending on the LPTMs of individual users within a WFL round. In the $t$-th WFL round, the WGPTM, denoted by $\Phi^t$, is the weighted sum of the LPTMs, $\Phi_{k}^t,\,\forall k\in \mathcal{K}$, as given by
	\begin{align}
		\Phi^t =\sum\limits _{k\in\mathcal{K}}e_{k}\Phi_{k}^t.\label{eq:U}
	\end{align}
Different from energy efficiency~\cite{9264742} or communication rate~\cite{9796982}, the WGPTM captures the impact of both communication and computation resources on the WFL convergence.
		
	\section{Proposed Joint Allocation of Power and Bandwidth in MC-NOMA WFL}\label{section:problem}
	In this paper, we aim to maximize the WGPTM of the MC-NOMA WFL per WFL round, by optimizing the power and bandwidth allocation in the round, subject to the delay, power, and bandwidth constraints.
	The problem is formulated as
	\begin{subequations}
		\begin{align}
			\text{\bf{P1}}:  &\mathop {\max }\limits_{\mathbf{p}^t_{n},B^t_{n},\,\forall n} \; \Phi^t   \\
			\mathrm{s.t.} \quad
			&\sum\limits_{n = 1}^N {{B^t_{n}}}  \leq  \;B, \label{constraints a} \\
			&{B^t_{n}}\;\ge 0 ,\forall n=1,\cdots,N,\label{constraints e}\\
			&0\; \le p^t_{n,i} \le  \;P_{\max },\forall n=1,\cdots,N, i=1, \cdots, |\mathcal{K}^t_{n}|,\label{constraints b}\\	
			&{T_{k}^t}  \le   \;T,\forall k \in {\mathcal {K},}\label{constraints c} \\
            &\text{\eqref{eq:T_u}},  \notag 
		\end{align}
	\end{subequations}
	where $P_{\max}$ is the maximum transmit power of a user, and $T$ is the duration of a WFL round. Constraint (\ref{constraints a}) ensures the total bandwidth budget is no wider than $B$. Constraint (\ref{constraints e}) ensures that the individual bandwidth is non-negative. \eqref{constraints b} specifies that  in the $t$-th WFL round, the transmit power of the $i$-th user in subchannel $n$ is non-negative and also must not exceed the maximum transmit power, $P_{\max }$. \eqref{constraints c} indicates that each user needs to complete a round of local model training and uploading within the round duration of $T$.
	Constraint \eqref{eq:T_u} specifies 
	the relationship between the transmit powers of the users, i.e., ${\mathbf{p}^t_{n}}$, and the bandwidths of the subchannels, i.e., ${B^t_{n}}$, according to the SIC decoding of NOMA.

	Problem \textbf{P1} is non-convex with respect to $\mathbf{p}_n^t$ and mathematically intractable. Specifically, ${\mathbf{p}}^t_{n}$ and $B^t_{n}$ are tightly coupled, and the objective $\Phi^t$ is non-convex with respect to ${\mathbf{p}^t_{n}}$. Moreover, constraint \eqref{eq:T_u} is concave in $p_{n,j},\forall j<i$, and convex in $p_{n,i}$. As a result, Problem {\bf{P1}} cannot be solved using existing optimization tools, e.g., CVX toolbox~\cite{cvx}.
	
	By substituting (\ref{eq:fl delay})--(\ref{iteration}) into \eqref{eq:U}, we rewrite the WGPTM in Problem \textbf{P1} as 
	\begin{align}
		\Phi^t =
		&
		\frac{\sum\limits_{k\in\mathcal{K}}\phi_{k}^t}{\underset{k\in\mathcal{K}}{\sum}|\mathcal{M}_{k}|}=
		\frac{\sum\limits _{k\in\mathcal{K}}T_{{\rm c},k}^t\beta_{k}}{\alpha\underset{k\in\mathcal{K}}{\sum}|\mathcal{M}_{k}|}=\frac{1}{{\alpha \sum\limits _{k\in\mathcal{K}}{\left| {{{\cal M}_k}} \right|} }}{\sum\limits _{k\in\mathcal{K}} {\left[ {({T_{k}^t} - T_{{\rm{u}},k}^t- T_{{\rm{d}}}^t)  {\beta_k}} \right]} }.\label{eq:U2}
	\end{align}
We notice that $\Phi^t\left(T_{k}^t\right)$ increases with $T_{k}^t$. Hence, $T_{k}^t=T$, $\forall k \in {\cal K}$, must be taken at the optimal solution to {\bf{P1}}. Since the downlink delay $T_{{\rm d}}^t$ is constant, Problem {\bf{P1}} can be rewritten as
		\begin{align}
			\notag\text{\bf{P2}}: &\quad\mathop {\min }\limits_{{\mathbf{p}^t_{n}},B^t_{n},\,\forall n} \sum\limits _{k\in\mathcal{K}} {T_{{\rm{u}},k}^t} {\beta_k}   
			\quad \mathrm{s.t.} \text{ \eqref{eq:T_u}, (\ref{constraints a}) 
   -- (\ref{constraints b})}.\notag
		\end{align}
Problem \textbf{P2} is still nonconvex and intractable. 
	
	In the rest of this section, we solve the problem optimally by first using variable substitution to 
	rewrite Problem \textbf{P2}, and then applying Cauthy's inequality to convert the rewritten problem equivalently to minimize the tight lower bound of its objective. The equality condition of Cauthy's inequality is exploited to establish the dependence of the variables at the optimum of the rewritten problem, thereby substantially reducing the number of variables per subchannel $n$ from $|\mathcal{K}_n|+1$ (corresponding to the bandwidth of the subchannel and the transmit powers of the $|\mathcal{K}_n|$ users)  to two (corresponding to the bandwidth and the transmit power of the user with the strongest channel gain). 
	The resultant problem with $2N$ variables is convex and can be efficiently solved using off-the-peg CVX toolboxes. For brevity of notation, we suppress the subscript ``$^t$'' below.

	\subsection{Variable Substitution}
	Define $W_{n}\left(\mathbf{p}_{n},B_{n}\right)=\sum\limits_{{k}\in {\cal K}_{n}} T_{{\rm{u}},{k}}  {\beta_{k}}$.  
	$\sum\limits _{k\in\mathcal{K}} { T_{{\rm{u}},{k}} {\beta_k}}=\sum\limits_{n= 1}^N W_{n}\left(\mathbf{p}_{n},B_{n}\right)$. Problem \textbf{P2} is recast as
		\begin{align}
			\notag	\text{\bf{P3}}:  &\mathop {\min }\limits_{{\mathbf {p}}_{n},B_{n},\,\forall n} \sum\limits_{n = 1}^N W_{n}\left(\mathbf{p}_{n},B_{n}\right) 			
			 \quad \mathrm{s.t. } \text{ \eqref{eq:T_u}, (\ref{constraints a}) -- (\ref{constraints b})}.  \notag
		\end{align}
	By substituting \eqref{eq:T_u}, $W_{n}\left(\mathbf{p}_{n},B_{n}\right)$ can be rewritten as
	\begin{align}
		W_{n}(&\mathbf{p}_{n},B_{n})	
		=\sum\limits_{{k}\in {\cal K}_{n}} T_{{\rm{u}},{k}}  {\beta_{k}}
		=  \sum\limits_{i = 1}^{|\mathcal{K}_{n}|} {S{\beta_{\mathcal{K}_{n}(i)}}}\Big/
  {{B_{n}{\log }_2}\bigg( {\frac{{\sum\limits_{j=1}^i {{{\left| {h_{n,j}} \right|}^2}p_{n,j} + {N_0}{B_{n}}} }}{{\sum\limits_{j = 1}^{i - 1} {{{\left| {h_{n,j}} \right|}^2}p_{n,j} + {N_0}{B_{n}}} }}} \bigg)}. \label{eq: W1_1}
	\end{align}
	To assess the monotonicity of $W_{n}\left(\mathbf{p}_{n},B_{n}\right)$ in $p_{n,|\mathcal{K}_{n}|}$, we reorganize \eqref{eq: W1_1} into
	{\small \begin{align}
	  \!\!\!   W_{n}(\mathbf{p}_{n},B_{n})=\sum\limits_{i = 1}^{|\mathcal{K}_{n}|-1} {\frac{{S{\beta_{\mathcal{K}_{n}(i)}}}}{{{B_{n}{\log }_2}\bigg( {\frac{{\sum\limits_{j=1}^i {{{\left| {h_{n,j}} \right|}^2}p_{n,j} + {N_0}{B_{n}}} }}{{\sum\limits_{j = 1}^{i - 1} {{{\left| {h_{n,j}} \right|}^2}p_{n,j} + {N_0}{B_{n}}} }}} \bigg)}}}+
	    \frac{{S{\beta_{\mathcal{K}_{n}(|\mathcal{K}_{n}|)}}}}{{{B_n{\log }_2}\bigg( {\frac{{{{\left| {h_{n,|\mathcal{K}_{n}|}} \right|}^2}p_{n,|\mathcal{K}_{n}|} +\sum\limits_{j = 1}^{|\mathcal{K}_{n}|-1} {{{\left| {h_{n,j}} \right|}^2}p_{n,j}+ {N_0}{B_{n}}} }}{{\sum\limits_{j = 1}^{|\mathcal{K}_{n}|-1} {{{\left| {h_{n,j}} \right|}^2}p_{n,j}+ {N_0}{B_{n}}} }}} \bigg)}}.\label{eq:W1}
	\end{align}}
	  We observe that $W_{n}\left(\mathbf{p}_{n},B_{n}\right)$ monotonically decreases with the increase of $p_{n,|\mathcal{K}_{n}|}$. Moreover, Problem \textbf{P3} is decoupled between $\mathbf{p}_n$ for $n=1,\cdots,N$ since the transmit powers of the users assigned to the different subchannels are independent.
	For this reason, the optimal transmit power of user $\mathcal{K}_{n}(|\mathcal{K}_{n}|)$, denoted by $p^*_{n,|\mathcal{K}_{n}|}$, which minimizes $W_{n}\left(\mathbf{p}_{n},B_{n}\right)$ takes its maximum $P_{\max}$: 
	\begin{align}
		p_{n,|\mathcal{K}_{n}|}^{* } = {P_{\max }}.\label{eq:pK}
	\end{align}
	
	To improve the tractability of (\ref{eq:W1}), we define $A_i\big(\mathbf{p}_n(i),B_n\big),\, i=0,1,\cdots, |\mathcal{K}_n|$ as
	\begin{equation}
		\label{eq:APi}
		A_{i}(\mathbf{p}_n(i),B_n)=\left\{
		\begin{aligned}
			&{{\log }_2}\left( {{N_0}{B_{n}}} \right),\quad\quad\quad\quad i =0;\\
			&\log_{2}\bigg(\stackrel[j=1]{i}{\sum}\left|h_{n,j}\right|^{2}p_{n,j}+N_{0}B_{n}\bigg), \quad\quad i=1,\cdots,|\mathcal{K}_{n}|,\\
		\end{aligned}
		\right.
	\end{equation}
	where $\mathbf{p}_n(i)$ 
	collects the first $i$ elements of $\mathbf{p}_n$, i.e., $\mathbf{p}_n(i)=\{p_{n,1},\cdots,p_{n,i}\}$ and $\mathbf{p}_n(0)=\varnothing$. 
	
	According to \eqref{eq:APi}, the transmit power of the $i$-th user assigned to the $n$-th subchannel, $p_{n,i}$, can be rewritten using $A_{i-1}(\mathbf{p}_n(i-1),B_n)$ and $A_{i}(\mathbf{p}_n(i),B_n)$, as given by
	\begin{align}
	    p_{n,i}=\frac{2^{A_{i}(\mathbf{p}_n(i),B_{n})}-2^{A_{i-1}(\mathbf{p}_n(i-1),B_{n})}}{|h_{n,i}|^2}, \quad  i=1,\cdots,|\mathcal{K}_n|.\label{p_ni}
	\end{align}
	
  	Based on \eqref{eq:APi}, we can establish the recurrence expression for  $A_{i}(\mathbf{p}_n(i),B_n)$, as given by
	\begin{subequations}\label{A_i,i-1}
	\begin{align}
		A_{i}(\mathbf{p}_n(i),B_n)&=\log_{2}\Big(\stackrel[j=1]{i}{\sum}\left|h_{n,j}\right|^{2}p_{n,j}+N_{0}B_{n}\Big) \label{eq: A_i,i-1,1}\\
		&=\log_{2}\left(\left|h_{n,i}\right|^{2}p_{n,i}+2^{A_{i-1}(\mathbf{p}_n(i-1),B_n)}\right),\label{eq: A_i,i-1,3}
	\end{align}
	\end{subequations}
	where $i=1,\cdots,|\mathcal{K}_{n}|$. Here, \eqref{eq: A_i,i-1,3} is based on the definition of $A_{i}(\mathbf{p}_n(i),B_n)$ in \eqref{eq:APi}.
	
	When $i=|\mathcal{K}_n|$, we substitute \eqref{eq:pK} into \eqref{A_i,i-1}, and obtain 
	\begin{align}
	 A_{|\mathcal{K}_{n}|}(\{\mathbf{p}_{n}(|\mathcal{K}_n|-1),p_{n,|\mathcal{K}_{n}|}^{* }\},B_{n})
	=	\log_2({\left| {h_{n,|\mathcal{K}_{n}|}} \right|}^2 P_{\max}+2^{A_{|\mathcal{K}_n|-1}(\mathbf{p}_{n}(|\mathcal{K}_n|-1),B_{n})}).\label{A_K}
	\end{align}
	
	For  brevity of notation, we define
	\begin{equation}
	 A_{i}(n)= A_{i}(\mathbf{p}_n(i),B_n),\,i=0,1,\cdots, |\mathcal{K}_{n}|-1,
	\end{equation}
	with $A_{|\mathcal{K}_{n}|}(n)$ given in \eqref{A_K}. 
	 $A_{|\mathcal{K}_{n}|}(n)$ depends on $A_{|\mathcal{K}_{n}|-1}(n)$.
  According to \eqref{eq:APi},  $	\forall i=1,\cdots,|\mathcal{K}_{n}|$,
	\begin{align}
	\log_2\Big(\big({\sum\limits_{j=1}^i {{{\left| {h_{n,j}} \right|}^2}p_{n,j} + {N_0}{B_{n}}} }\big)\Big/\big({\sum\limits_{j = 1}^{i - 1} {{{\left| {h_{n,j}} \right|}^2}p_{n,j} + {N_0}{B_{n}}} }\big)\Big)&=A_{i}(n)-A_{i-1}(n).
	\label{A-A}
	\end{align}	
	By substituting \eqref{A-A} into \eqref{eq:W1}, 
	$W_n(\mathbf{p}_n,B_n)$ is rewritten as
	\begin{align}
		\notag  W_{n}\Big(\big\{\mathbf{p}_{n}(|\mathcal{K}_{n}|-1),p^*_{n,|\mathcal{K}_{n}|}\big\},B_n\Big)
		&=\underset{=w_1^{n}}{\underbrace{\frac{S}{{{B_{n}}}} \frac{{{\beta_{\mathcal{K}_{n}(|\mathcal{K}_{n}|)}}}}{{	A_{|\mathcal{K}_{n}|}(n)- A_{|\mathcal{K}_{n}|-1}(n)}}}}
		+\underset{=z_1^{n}}{\underbrace{\frac{S}{B_{n}}  \sum\limits_{i= 1}^{|\mathcal{K}_{n}|-1} {\frac{\beta_{\mathcal{K}_n(i)}}{{A_i(n)  - A_{i-1}(n) }}} }}\\
  &\triangleq W_{n}\left({\cal A}_n,B_{n}\right),\label{eq:W2}
	\end{align}
where 
${\cal A}_n=\{ A_1(n), \cdots,A_{|\mathcal{K}_n|-1}(n)\}$.
Moreover, $w_1^{n}$ and $z_1^{n}$ are defined for brevity of notation. As shown, $w_1^{n}$ depends on $A_{|\mathcal{K}_{n}|-1}(n)$ and $B_{n}$, while $z_1^{n}$ depends on 
${\cal A}_n$ and $B_{n}$.

	As a result, Problem \textbf{P3} can be rewritten as
	\begin{align}
			\notag	\text{\bf{P4}}:  &\quad\mathop {\min }\limits_{{\cal A}_{n},B_{n},\forall n }\sum\limits_{n = 1}^N W_{n}\left({\cal A}_n,B_{n}\right) \quad
			\mathrm{s.t.}
			\quad\text{ (\ref{constraints a}) -- (\ref{constraints b}), \eqref{eq:APi}, \eqref{eq:W2}}, \notag
		\end{align}
 where the newly added constraint \eqref{eq:APi} preserves the mapping between $p_n$ and $\mathcal{A}_n$ during the variable substitution, and \eqref{eq:W2} specifies $W_n({\cal A}_n,B_{n})$. 
The objective of Problem \textbf{P4} is non-convex in $\mathcal{A}_n$ since the $i$-th element in $z_1^n$, i.e., $\frac{\beta_{\mathcal{K}_n(i)}}{{A_i(n)  - A_{i-1}(n) }}$, is concave in $ A_{i-1}(n)$ and convex in $A_i(n)$. Moreover, constraints~\eqref{constraints b} and~\eqref{eq:APi} are intractable as $\mathcal{A}_n$ is not a liner mapping of~$\mathbf{p}_n$.
	
	\subsection{Variable Reduction by Cauchy's Inequality}
	We proceed to reduce the number of optimization variables and, in turn, improve the tractability of Problem \textbf{P4}. To do this, we first provide the following lemma.
	\begin{proposition}		
	The following inequality always holds
			\begin{align}
				\sum\limits _{i=1}^{|\mathcal{K}_{n}|-1}\frac{\beta_{\mathcal{K}_{n}(i)}}{A_i({n})-A_{i-1}({n})}
				\ge\bigg(\sum\limits _{i=1}^{|\mathcal{K}_{n}|-1}\sqrt{\beta_{\mathcal{K}_{n}(i)}}\bigg)^{2}\bigg/\bigg(A_{|\mathcal{K}_{n}|-1}({n})-A_0({n})\bigg),\label{eq:cauchy2}
			\end{align}
		where the equality is taken only when  ${\cal A}_n$ satisfies
		\begin{align}
				 \frac{A_{|\mathcal{K}_{n}|-1}({n})-A_0({n})}{\sum\limits _{i=1}^{|\mathcal{K}_{n}|-1}\sqrt{\beta_{\mathcal{K}_{n}(i)}}}
			= \frac{{A_i({n})-A_{i-1}({n})}}{{\sqrt {{\beta_{\mathcal{K}_{n}(i)}}} }} \triangleq Q_n, 
			i=1,\cdots,|\mathcal{K}_{n}|-1.\label{equality conditions}
		\end{align}
	\end{proposition}
	\begin{proof}\let\qed\relax
		Based on Cauchy's inequality, we have
		\begin{align}
			\bigg(\sum\limits _{i=1}^{|\mathcal{K}_{n}|-1}\left(a_{i}\right)^{2}\bigg)\bigg(\sum\limits _{i=1}^{|\mathcal{K}_{n}|-1}\left(b_{i}\right)^{2}\bigg)\geq\bigg(\sum\limits _{i=1}^{|\mathcal{K}_{n}|-1}a_{i}b_{i}\bigg)^{2},\label{cauchy}
		\end{align}
		which takes equality when $\frac{a_1}{b_1}=\cdots = \frac{a_{|\mathcal{K}_n|-1}}{b_{|\mathcal{K}_n|-1}}=\frac{\sum_{i=1}^{|\mathcal{K}_n|-1} a_i}{\sum_{i=1}^{|\mathcal{K}_n|-1} b_i}$.
		
		Since 
		$	a_i=\sqrt{{A_i({n})-A_{i-1}({n})} }$, then $\sum\limits _{i=1}^{|\mathcal{K}_{n}|-1}\left(a_{i}\right)^{2}    =\sum\limits _{i=1}^{|\mathcal{K}_{n}|-1}A_i(n)-\sum\limits _{i=1}^{|\mathcal{K}_{n}|-1}A_{i-1}(n)=A_{|\mathcal{K}_{n}|-1}(n)-A_0(n)$. Also recall 
		$	b_i=\sqrt{\frac{\beta_{\mathcal{K}_{n}(i)}}{{A_i({n})-A_{i-1}({n})}}}$. 
		Then, (\ref{cauchy}) can be rewritten as 
		\begin{align}
			\left(A_{|\mathcal{K}_{n}|-1}\left(n\right)-A_0(n)\right)\times&\sum\limits _{i=1}^{|\mathcal{K}_{n}|-1}\frac{\beta_{\mathcal{K}_n(i)}}{A_i(n)-A_{i-1}(n)}
			\ge\bigg(\sum\limits _{i=1}^{|\mathcal{K}_{n}|-1}\sqrt{\beta_{\mathcal{K}_n(i)}}\bigg)^{2},\label{cauchy3}
		\end{align}
		where the equality is taken in \eqref{cauchy3} only when \eqref{equality conditions} is satisfied. 
		By dividing both sides of (\ref{cauchy3}) by $\left(A_{|\mathcal{K}_{n}|-1}\left({n}\right)-A_0({n})\right)$, (\ref{eq:cauchy2}) is finally obtained and the lemma is proved.
		$\hfill\blacksquare$
	\end{proof}
	
	By multiplying both sides of \eqref{eq:cauchy2} by $\frac{S}{B_n}$, we have
	\begin{align}
			z_1^n=&	\frac{S}{B_n}\sum\limits _{i=1}^{|\mathcal{K}_{n}|-1}\frac{\beta_{\mathcal{K}_{n}(i)}}{A_i({n})-A_{i-1}({n})}
		\ge \frac{S}{B_n}\frac{\Big(\sum\limits_{i=1}^{|\mathcal{K}_{n}|-1}\sqrt{\beta_{\mathcal{K}_{n}(i)}}\Big)^{2}}{A_{|\mathcal{K}_{n}|-1}({n})-A_0({n})}\triangleq z_2^n.\label{Z_2}
	\end{align}
	By substituting \eqref{Z_2} into (\ref{eq:W2}), the minimum of $W_{n}$ can be obtained, as given by
	\begin{align}
		{W_{n}}\left({\cal A}_{n},B_{n}\right) & =	w_1^{n}+z_{1}^{n}\ge 	w_1^{n}+z_{2}^{n}
		\triangleq Z_n\left({A_{|\mathcal{K}_{n}|-1}({n})},B_{n}\right),
		\label{eq:W3}
	\end{align}
	where the equality is taken when (\ref{equality conditions}) holds. In other words,  $Z_n\left({A_{|\mathcal{K}_{n}|-1}({n})},B_{n}\right)$ is the minimum of ${W_{n}}\left({\cal A}_{n},B_{n}\right) $, and is taken when \eqref{equality conditions} holds. By reorganizing \eqref{equality conditions},  at the minimum of ${W_{n}}\left({\cal A}_{n},B_{n}\right) $, $A_i({n})$ can be rewritten as 
	\begin{subequations}
	   \begin{align}  A_{i}({n}) &= 
		{\sum\limits_{j = 1}^i {\sqrt {{\beta_{\mathcal{K}_n(j)}}} }{Q_n} }+A_{0}(n)\label{A_Q}\\
		=	&{\sum\limits_{j = 1}^i {\sqrt {{\beta_{\mathcal{K}_n(j)}}} }{\frac{A_{|\mathcal{K}_{n}|-1}({n})-A_0({n})}{\sum\limits _{i=1}^{|\mathcal{K}_{n}|-1}\sqrt{\beta_{\mathcal{K}_{n}(i)}}}} }+A_{0}({n}).		\label{A_ii}
	\end{align}
	\end{subequations}
	Here, $A_i({n}),\forall i=1,\cdots,|\mathcal{K}_n|-1$, is expressed as a function of only $A_{|\mathcal{K}_{n}|-1}({n})$. We can evaluate the optimal values of all $A_i({n}),\forall i=1,\cdots,|\mathcal{K}_n|-1$, provided the optimum of $A_{|\mathcal{K}_{n}|-1}({n})$ avails.
	
	We sum up (\ref{eq:W3}) for all subchannels $n=1,\cdots, N$. Then,
	\begin{align}
		\sum\limits_{n = 1}^N {{W_{n}}}\left({\cal A}_n,B_{n}\right)   \ge \sum\limits_{n = 1}^N Z_n\left(A_{|\mathcal{K}_{n}|-1}({n}),B_{n}\right),\label{eq:R17}
	\end{align}
	where the equality is taken when (\ref{equality conditions}) is satisfied $\forall n$. 
	In other words, $\sum\limits_{n = 1}^N Z_n\left(A_{|\mathcal{K}_{n}|-1}({n}),B_{n}\right)$ is the minimum of $\sum\limits_{n = 1}^N {{W_{n}}}\left({\cal A}_n,B_{n}\right)$ and is taken when \eqref{A_ii} holds $\forall i\leq |\mathcal{K}_n|-1, \forall n$.  

Finally, Problem \textbf{P4} is rewritten as  
	\begin{subequations}
		\begin{align}
			\text{{\bf{P5}}}:& \quad\mathop{\min}   \limits_{{A_{|\mathcal{K}_{n}|-1}({n}),{B_{n}},\forall n}} \sum\limits_{n = 1}^N {{Z_n}\left(A_{|\mathcal{K}_{n}|-1}({n}),B_{n}\right)}  \\
			\mathrm{s.t.}
			  \quad &A_{|\mathcal{K}_{n}|-1}({n}) \geq \log_{2}\left(N_{0}B_{n}\right),\,\forall n=1,\cdots,N, \label{p5:constraints1}\\
			&A_{|\mathcal{K}_{n}|-1}({n}) \leq \log_{2}\bigg(\stackrel[j=1]{|\mathcal{K}_{n}|-1}{\sum}\left|h_{n,j}\right|^{2}P_{\max}+N_{0}B_{n}\bigg),\forall n=1,\cdots,N,\label{p5:constraints3}\\
			&Z_n\left({A_{|\mathcal{K}_{n}|-1}({n})},B_{n}\right)=	w_1^{n}+z_{2}^{n},\label{p5:constraint4}\\
   \quad&\text{(\ref{constraints a}), \eqref{constraints e}},\notag
		\end{align}
	\end{subequations}
	where \eqref{p5:constraints1} provides the lower bound of $A_{|\mathcal{K}_{n}|-1}({n})$ and obtained by substituting \eqref{eq:APi} into $p_{n,i}>0,\, \forall i=1,\cdots,|{\cal K}_n|-1$ in \eqref{constraints b}. Constraint {\eqref{p5:constraints3} is the upper bound of $A_{|\mathcal{K}_{n}|-1}({n})$} and is obtained by substituting \eqref{eq:APi}  into  $p_{n,i}<P_{\max},\,\forall i=1,\cdots,|{\cal K}_n|-1$ in  \eqref{constraints b}.  Constraint \eqref{p5:constraint4} is equivalent to \eqref{eq:W2} since $W_{n}\left({\cal A}_n,B_{n}\right)$ and
	$Z_n(A_{|\mathcal{K}_{n}|-1}({n}),B_{n})$ take the same minimum.
		
	\begin{proposition}$Z_n(A_{|\mathcal{K}_{n}|-1}({n}),B_{n})=w_1^{n}+z_2^{n}$ is convex in both $A_{|\mathcal{K}_{n}|-1}({n})$ and $B_{n}$ due to the positive semi-definite Hessian matrices of both $w_1^{n}$ and $z_2^{n}$.
	\end{proposition}
	\begin{proof}\let\qed\relax
		Please see Appendix~\ref{proof-lemma2}. 
	\end{proof}
	
	The objective of Problem \textbf{P5}, $\sum\limits_{n = 1}^N {{Z_n}\left(A_{|\mathcal{K}_{n}|-1}({n}),B_{n}\right)}$, is convex
	with respect to $A_{|\mathcal{K}_{n}|-1}({n})$ and $B_{n}$, since ${{Z_n}\left(A_{|\mathcal{K}_{n}|-1}({n}),B_{n}\right)}$ is convex as stated in \textbf{Lemma 2} and a non-negatively weighted sum of convex functions is convex \cite{boyd2004convex}.
	Moreover, the constraints of Problem {\bf{P5}}, i.e., \eqref{constraints a}, \eqref{constraints e}, \eqref{eq:W3}, \eqref{p5:constraints1}, and \eqref{p5:constraints3}, are convex or linear. All the objective and constraints are differentiable.
	As a result, Problem {\bf{P5}} is convex. The optimal solution to the problem, $A^*_{|\mathcal{K}_{n}|-1}({n})$ and $B_n^{*}$, can be efficiently solved by off-the-peg CVX tools, e.g., interior point method~\cite{cvx}. 
Moreover, there are $2N$ variables in Problem \textbf{P5},  substantially fewer than $K$ in Problem \textbf{P4} (including $N$ variables of $B_n, \forall n$ and $(K-N)$ variables of $\mathcal{A}_n,\forall n$), since ${\cal A}_{n}=\{A_1(n),\cdots,A_{|\mathcal{K}_{n}|-1}(n)\}$ and $K=\sum_{n=1}^N |\mathcal{K}_{n}| \gg N$. This contributes considerably to the simplification of the considered problem. Let alone Problem \textbf{P4} is non-convex and intractable.
	
	By substituting $A^*_{|\mathcal{K}_{n}|-1}({n})$ and $B^{*}_n$ into 
	(\ref{equality conditions}), we can obtain the optimum of $Q_n$, i.e., $Q_n^*$, as
	\begin{align}
		Q_n^*={\left(A^*_{|\mathcal{K}_{n}|-1}({n})-A_0^*(n)\right)}\Big/{\Big(\sum_{i=1}^{|\mathcal{K}_{n}|-1}\sqrt{\beta_{\mathcal{K}_{n}(i)}}\Big)},\label{Q_n}
	\end{align}
	where $A_0^*(n)=\log_2(N_0B^*_n)$; see \eqref{eq:APi}.
	According to \eqref{A_Q}, with $A_0^*(n)$, $A^*_{|\mathcal{K}_{n}|-1}({n})$ and $Q_n^*$, we can obtain the optimal $A_{i}({n})$, denoted by $A^*_{i}({n})$, as   
	\begin{align}  A^*_{i}({n}) = &
		{\sum\limits_{j = 1}^i {\sqrt {{\beta_{\mathcal{K}_n(j)}}} }{Q_n^{*}} }+A_{0}^*({n}),
		\forall i \leq |\mathcal{K}_n|-1,\forall
		 n.\label{A_i}
	\end{align}	

	Finally, we can obtain the optimal transmit powers of the users, denoted by ${\bf{p}}_n^{*},\forall n$, by substituting \eqref{A_i} into 
	\eqref{p_ni}, and ${\bf{p}}_n^{*}=\big\{ p_{n,i}^{*}, \forall i=1,\cdots,|\mathcal{K}_n|\big\}$ is given by 
	\begin{align}\label{P_i}
		p_{n,i}^{*}  = \left\{ {\begin{array}{*{20}{c}}
				\min \Bigg\{ {\frac{{{2^{\sum\limits_{j= 1}^i {\sqrt {{\beta_{\mathcal{K}_n(j)}}} }{Q_n^{*}} }} - {2^{\sum\limits_{j = 1}^{i - 1} {\sqrt {\beta_{\mathcal{K}_n(j)}} }{Q_n^{*}} }}}}{{{{\left| {{h_{n,i}}} \right|}^2}}},P_{\max }} \Bigg\},	 i=1,\cdots,|\mathcal{K}_n|-1;	\\
				{P_{\max }},\hfill \quad i=|\mathcal{K}_n|.
		\end{array}} \right.
	\end{align}
	The optimal $ T_{{\rm{u}},{k}} $ and $ T_{{\rm{c}},{k}} $, denoted by $ T^*_{{\rm{u}},{k}} $ and $ T^*_{{\rm{c}},{k}} $, can be evaluated according to ${\bf{p}}_n^{*}$ and $B^{*}_n$. Specifically, $T^*_{{\rm{u}},{k}}$ can be obtained by substituting ${\bf{p}}_n^{*}$ and $B^{*}_n$ into \eqref{eq:T_u}. Then, $ T^*_{{\rm{c}},{k}}=T-T^*_{{\rm{u}},{k}}-T_{\rm{d}}$. 
	In addition to allocating the computing and communication resources, we can also allocate
	the optimal LPTM, $\Phi^*_{k}$, by substituting $T^*_{{\rm{c}},{k}}$ into \eqref{eq:fl delay} and \eqref{iteration}.
	
	\subsection{Algorithm Description and Discussion}
	
	\textbf{Algorithm \ref{Algorithm 2}} summarizes the proposed joint power and bandwidth allocation in the new MC-NOMA WFL system. The complexity of the algorithm is dominated by the convex optimization used to solve Problem \textbf{P5}, since the rest of the algorithm, including \eqref{Q_n}--\eqref{P_i} in Step~4, \eqref{eq:T_u} in Step 5, and \eqref{eq:fl delay} and \eqref{iteration} in Step 7, only involve arithmetic operations and are relatively computationally negligible.
	The worst-case complexity of using a typical convex optimization solver, e.g., interior point method, to solve Problem \textbf{P5} is $\mathcal{O}(\max\{N_c,N_v\}^4\sqrt{N_v})\log(\frac{1}{\epsilon})$~\cite{9387137}, where $N_c$ and $N_v$ are the numbers of constraints and variables, respectively; and $\epsilon$ is the convergence accuracy of the algorithm. For Problem \textbf{P5}, $N_c=3N+1$, accounting for $N$ linear constraints in \eqref{p5:constraints1}, $N$ linear constraints in \eqref{p5:constraints3}, and $N+1$ linear constraints in \eqref{p5:constraint4}. $N_v=2N$, accounting for the $2N$ variables, i.e., $A_{|\mathcal{K}_{n}|-1}({n})$ and ${B_{n}},\forall n=1,\cdots,N$. As a result, the overall complexity of \textbf{Algorithm \ref{Algorithm 2}} is $\mathcal{O}\big(N^{4.5}\log(\frac{1}{\epsilon})\big)$.

	\begin{algorithm}[h]
		\caption{The optimal allocation of MC-NOMA WFL with Flexible Aggregation}
		\label{Algorithm 2}
		\textbf{Initialization:} $\mathcal{K}$, $N$, $B$, ${\cal D} _k$, $T$, $T_{\rm d}$, $P_{\max}$, $\alpha$, $\beta_k$.
		
		\textbf{For any WFL round, i.e., the $t$-th WFL round:}
		
	   \quad \textbf{\% Allocation of power and bandwidth: $p_{n,i}^{* }$ and $B_n^{*}$} 
		
		\begin{algorithmic}[1]
		    \STATE Collect the channel gains of all users in all subchannels; 
			\STATE Assign $K$ users to $N$ subchannels, and obtain $\mathcal{K}_{n}$ and $h_{n,i}$, $\forall i =1,\cdots,|\mathcal{K}_n|,\forall
		 n=1,\cdots,N$;
			\STATE  Solve Problem \textbf{P5}  to obtain $A^*_{|\mathcal{K}_{n}|-1}({n})$ and $B_n^{*}$ with CVX tools~\cite{cvx};
            \STATE Obtain $Q_n^*,A^*_{i}({n})$, and $p_{n,i}^{* }$ using \eqref{Q_n}, \eqref{A_i}, and \eqref{P_i}, respectively.
		
     	\textbf{\% Allocation of communication tasks: $\Phi^*_{k}$}
		
			\STATE Obtain $T^*_{{\rm{u}},{k}}$ by substituting ${\bf{p}}_n^{*}$ and $B^{*}_n$ into \eqref{eq:T_u}; 
			\STATE Evaluate $ T^*_{{\rm{c}},{k}}=T-T^*_{{\rm{u}},{k}}-T_{\rm{d}}$;
			\STATE Obtain $\Phi^*_{k}$ by substituting $T^*_{{\rm{c}},{k}}$ into \eqref{eq:fl delay} and \eqref{iteration}.
		\end{algorithmic}  
   \textbf{Output:} $p_{n,i}^{* }$, $B_n^{*}$, and $\Phi^*_{k}$.
	\end{algorithm}
	
\section{Numerical Results}\label{section: simulation}
In this section, extensive simulations are carried out to test the new MC-NOMA WFL with Flexible Aggregation and the proposed \textbf{Algorithm \ref{Algorithm 2}} on image classification tasks. 
\vspace{-3mm }
\subsection{Compared Algorithms}\label{setups}
As discussed in Section~\ref{convergence}, there is no directly comparable technique in the literature, due to the new design of MC-NOMA WFL with Flexible Aggregation.
We test several alternatives to the proposed \textbf{Algorithm \ref{Algorithm 2}} to assess the gains of MC-NOMA, Flexible Aggregation, and joint power and bandwidth allocation.
Under the Flexible Aggregation, we implement:
\begin{itemize}
    \item 
    \textbf{MC-NOMA with optimal joint power and bandwidth allocation (Joint Allocation):} 
    This is the case where \textbf{Algorithm \ref{Algorithm 2}} is run to maximize the WGPTM for each model aggregation. 
    
    \item
    \textbf{MC-NOMA with optimal power and equal bandwidth allocation (Power-only Allocation):}
    By setting each subchannel with equal bandwidth $\frac{B}{N}$, \textbf{Algorithm \ref{Algorithm 2}} is used to optimize the transmit powers of the users only. 
    
    \item 
    \textbf{MC-NOMA with full power and equal bandwidth (Full Power):} All users transmit with their full power $P_{\max}$ in subchannels with equal bandwidth $\frac{B}{N}$.

    \item 
    \textbf{MC-OMA/TDMA with optimal time and equal bandwidth allocation (MC-OMA):} 
    The users are assigned to subchannels of equal bandwidth. In a subchannel, the users transmit the maximum power $P_{\max}$ in non-overlapping time slots. The duration of the time slots can be evaluated. The transmission order of the users is enumerated to maximize the WGPTM.     
\end{itemize}

We also implement these algorithms under the extensively considered Sync-FL where all users run an equal number of iterations before a global aggregation~\cite{pmlr-v54-mcmahan17a,9264742}.
\begin{itemize}
    \item \textbf{Joint Allocation}: By updating Problem \textbf{P1} to be a max-min objective, the optimal power and bandwidth allocation is obtained under Sync-FL using a slightly modified version of the proposed \textbf{Algorithm \ref{Algorithm 2}}; see Appendix C. 
    
    \item \textbf{Power-only Allocation: } By setting the bandwidth of each subchannel to $\frac{B}{N}$, the optimal power allocation is readily evaluated; see Appendix~C. 
    
    \item \textbf{Full Power: }  
    The key difference between this scheme and its counterpart under Flexible Aggregation is that after the transmit powers and subchannel bandwidths are allocated, the smallest number of iterations among all users is identified. All users upload their local models trained on the smallest number of iterations.
    
    \item \textbf{MC-OMA: } The key difference between this scheme and its counterpart under Flexible Aggregation is that the transmission order of the users is enumerated to maximize the minimum of the LPTMs per subchannel. Then, the minimum LPTM of all subchannels specifies the number of iterations for every user per WFL round. 
\end{itemize}

Additionally, we consider MC-OMA/TDMA for Async-FL based on~\cite{TT-fed}, where each user completes 4 iterations before uploading its local models. Different from~\cite{fedasync,Fed-AT,TT-fed,9725259}, the channel and computing power of a user change over time in this paper. It is nontrivial to implement Async-FL under MC-OMA/TDMA.
We consider such an ideal setting that in each subchannel, the total of the model uploading times of all users and the training time of the computationally slowest user for 4 iterations is set to be an aggregation round of the subchannel. (The users transmit full powers and their model uploading times can be obtained, as under Sync-FL.) The longest round of all subchannels is a global aggregation round. Other users can train \textbf{$4X$} ($X \in \mathbb{N}$) iterations during a global round. We assume ideally each user has its local model aggregated superficially every 4 iterations with others' models that are also due to aggregate, with no additional model uploading delays, between two global aggregations involving the slowest user. 
As such, we can obtain better learning accuracy and converge faster than the practical Async-FL.

\subsection{Simulation Setup}

We set the bandwidth of the system, i.e., $B$, to 30 MHz. The bandwidth is evenly divided into $N$ subchannels. By default, $N=10$. The maximum transmit power of a user is $P_{\max}=46\ \text{dBm}$. The normalized channel gain of the users, i.e., $\frac{h}{\sqrt{N_0}} $, follow the uniform distribution between $2$ and $15$ dB, which changes independently between WFL rounds~\cite{7557079}. The processing speed of user $k$, i.e., $\beta_k$, obeys the uniform distribution between $6$ and $9$ GFLOPS, which can also change independently between WFL rounds due to random background traffic and operations.

 By default, we cluster all $K$ users for all $N$ subchannels by sorting the users in the ascending order of their channel gains and labeling them from 1 to $N$. The users with the same label are assigned to a subchannel, as done in~\cite{7557079}. Alternatively, we cluster the users at random by permuting them randomly and labeling them from 1 to $N$. The users with the same label are assigned to a subchannel. Note that NOMA user clustering is beyond the scope of this paper, and the proposed \textbf{Algorithm~\ref{Algorithm 2}} applies to any clustering method.

The simulations are performed based on the open-source FL model, namely, FedML~\cite{chaoyanghe2020fedml}.  
Two image classification datasets, i.e., F-MNIST and F-CIFAR100, are considered. 
\begin{itemize}
    \item The F-MNIST dataset~\cite{AdaptiveFL} extends the standard MNIST dataset of handwritten digits by adding upper- and lower-case English characters. The dataset contains 62 labels with 671,585 training examples and 77,483 test examples. 
    Each example is a $28\times28$-pixel square grayscale image. The F-MNIST dataset is partitioned into 3,400 i.i.d. training and test sets.

    \hspace{3 mm} We also extend the F-MNIST dataset to construct two non-i.i.d. datasets, referred to as NON-IID-2 and NON-IID-4.
    In the NON-IID-2 dataset, the 62 labels are divided into two groups with 31 labels each. Each user randomly chooses a group and 31 labels uniformly randomly from the group as its training and test examples. In the NON-IID-4 dataset, the 62 labels are divided into four groups. Each client chooses a group and 15 or 16 labels uniformly randomly from the group as its training and test examples.
    
    \hspace{3 mm} We consider a CNN model to train the F-MNIST dataset. The CNN model consists of two convolutional (Conv) layers using 32 and 64 convolutional filters with the size of $3 \times 3$ each. We insert a pooling layer using max-pooling with $2\times 2$ receptive fields between two consecutive Conv layers to prevent overfitting. We also insert two fully-connected (FC) layers to integrate the local information of the previous layer, i.e., the Conv or max-pooling layer. The last FC layer holds the output. Each Conv layer and FC layer are followed by a ReLU activation function, called a ReLU layer.

    \hspace{3 mm}When training the CNN on the F-MNIST dataset, the size of the model parameter is $S=4.84$~Mbytes. The number of FLOPs to train a sample is $\alpha=0.04$ GFLOPs. Here, $S$ and $\alpha$ are obtained by calling the function \textit{get\_model\_complexity\_info()} in the Python project \textit{ptflops}. The local dataset size $|\mathcal{D}_k|$ obeys the uniform distribution between 300 and 500. The duration of a round is $T=10$ seconds. The mini-batch size is set to $20$.  
    
    \item The F-CIFAR100 dataset assigns the CIFAR100 dataset to 500 training and 100 test sets. The CIFAR100 dataset contains 100 classes with 600 images (including 500 training examples and 100 test examples)  per class. Each image is a $32\times 32$-pixel color image.

\hspace{3 mm} We consider a ResNet18 model to train the F-CIFAR100 dataset. The ResNet18 has 18 parameter layers, including 17 Conv layers and an FC layer. Compared to a plain CNN network, shortcut connections are inserted to create a corresponding residual network~\cite{He_2016_CVPR}.

\hspace{3 mm} When training the ResNet18 on the F-CIFAR100 dataset, the size of the model parameter is $S=46.76$~Mbytes and the number of FLOPs to train a sample is $\alpha=0.08$ GFLOPs,  obtained by calling the function \textit{get\_model\_complexity\_info()} in the Python project \textit{ptflops}. The local dataset size $|\mathcal{D}_k|$ is $100$. A round lasts $T=30$ seconds. The mini-batch size is~10.

\end{itemize}
All simulations are conducted on a server with Intel(R) Xeon(R) CPU E5-2628 v3@2.50GHz and 126G memory and a GeForce GTX 1080 Ti GPU with 11178MiB memory, running Python 3.7.11, Numpy 1.21.2, and PyTorch `1.10.1' installed on an Ubuntu 18.04.5 LTS system.

\subsection{Simulation Results}

Fig. \ref{fig:dogl} evaluates the scalability of the proposed \textbf{Algorithm \ref{Algorithm 2}} by plotting the WGPTM with the increase of users, subchannels, and round duration, where the default NOMA user clustering method developed in~\cite{7557079} is adopted. For comparison, it also shows the alternative algorithms described in Section~\ref{setups}.
Each result is the average of 1,000 independent simulations.

\begin{figure}
	\centering
	\subfigure[$N=10$ and $T=10$ s.]
	{
		\begin{minipage}[t]{0.3\textwidth}
			\centering
			\includegraphics[width=5.3cm]{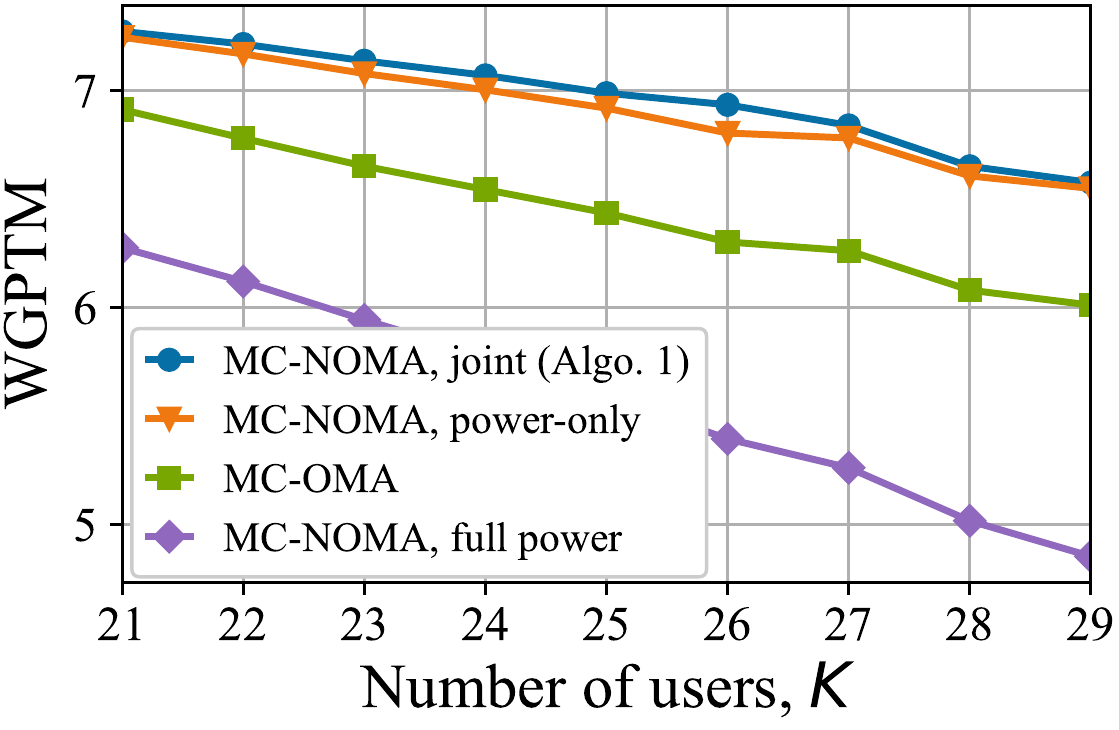}
		\end{minipage}
	}
	\subfigure[$K=25$ and $T=10$ s. ]
	{
		\begin{minipage}[t]{0.3\textwidth}
			\centering
			\includegraphics[width=5.3cm]{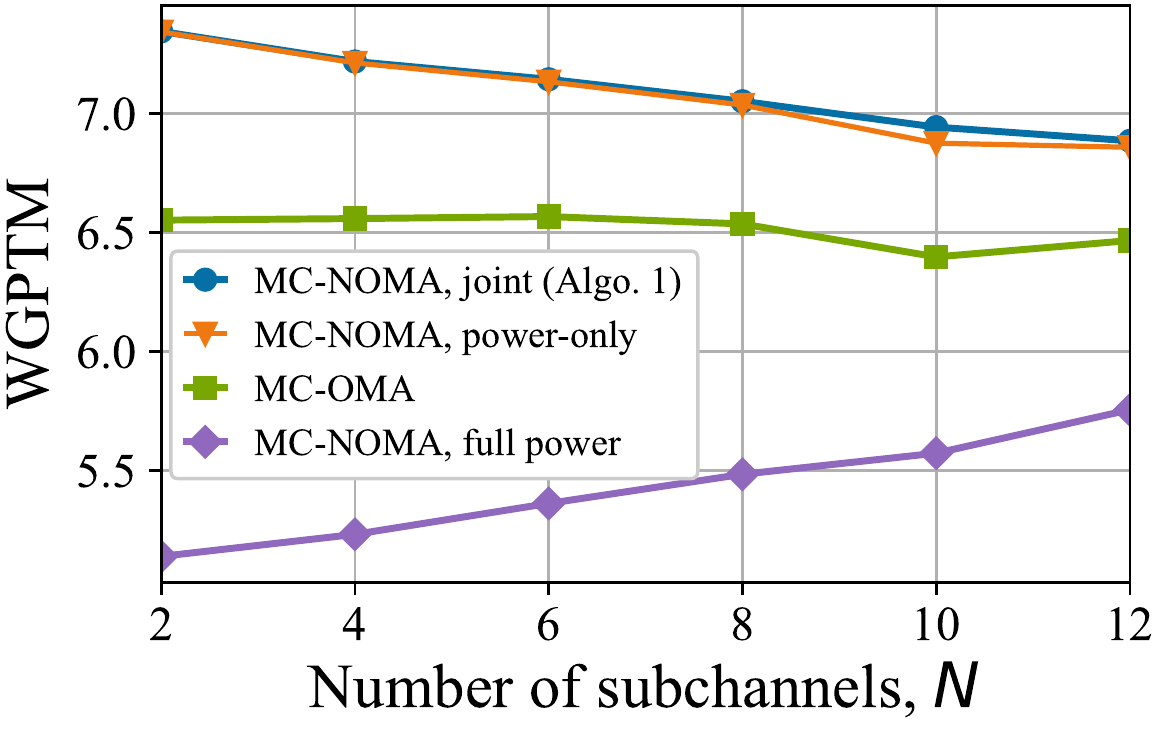}
		\end{minipage}
	}
	\subfigure[$K=25$ and $N=10$.]
	{
		\begin{minipage}[t]{0.3\textwidth}
			\centering
			\includegraphics[width=5.3cm]{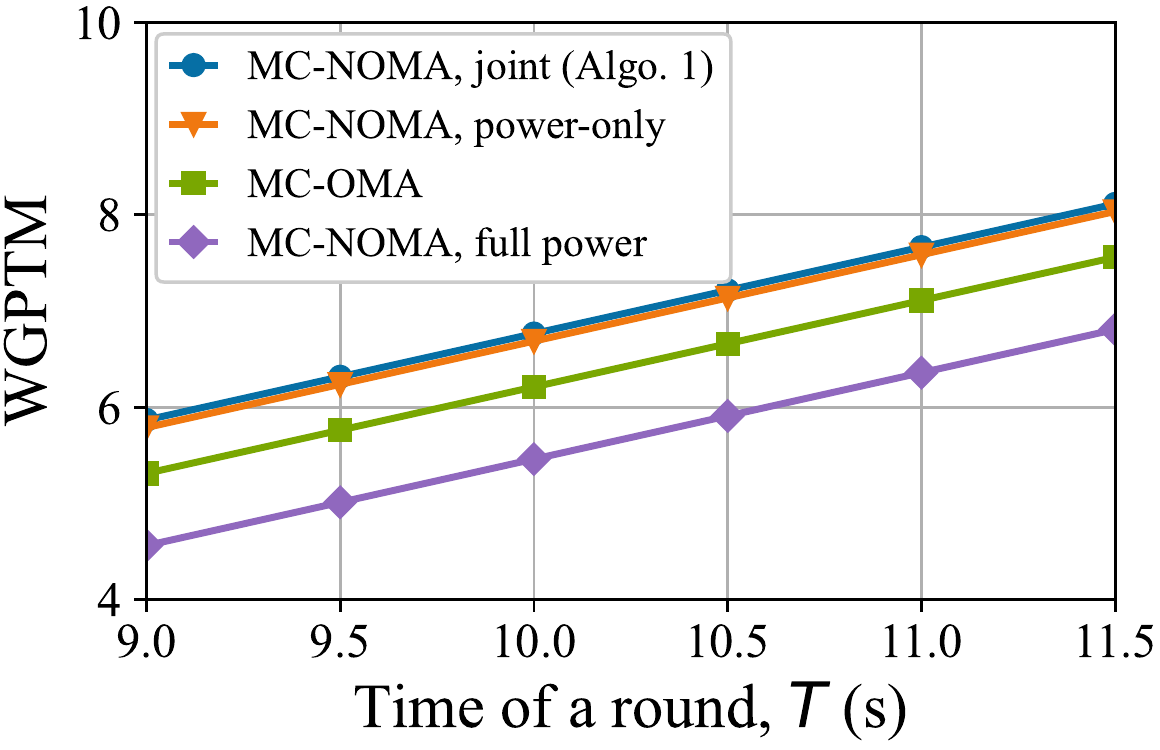}
		\end{minipage}
	}\caption{The WGPTM of the MC-NOMA WFL versus (a) the number of users, (b) the number of subchannels, and (c) the duration of a global round, where the NOMA user clustering method developed in~\cite{7557079} is adopted. The CNN model and F-MNIST dataset are considered. 
	} 
	\label{fig:dogl}
\end{figure}

Fig. \ref{fig:dogl}(a) shows 
that the WGPTM declines under all the considered algorithms with the increase of users, where the system bandwidth is $B=30$ MHz and the number of subchannels is $N=10$. The reason is that the growing number of users increasingly intensifies the co-channel interference in each subchannel, 
hence extending their transmit time, shortening their training time, and reducing the number of mini-batches trained in a WFL round. 
We also see that optimal joint power and bandwidth allocation under Flexible Aggregation (i.e., \textbf{Algorithm \ref{Algorithm 2}})
outperforms all other algorithms. 
The gain of NOMA over OMA is evident.

Fig.~\ref{fig:dogl}(b) plots the WGPTM with an increasing number of subchannels, where the total system bandwidth remains $B=30$ MHz. 
We see that the proposed \textbf{Algorithm \ref{Algorithm 2}} demonstrates significant improvement over its alternatives, especially when the number of subchannels is small. The reason is that more users have to share a subchannel when the number of subchannels is smaller, in which case meticulous allocation of the transmit powers and subchannel bandwidths plays a critical role. 
We also see that the WGPTMs of the considered algorithms, except one, 
decrease with the growing number of subchannels. The reason is that it takes longer for each user to upload its local model when the bandwidth of a subchannel is narrower, leaving a shorter computing period within a WFL round. The only exception is the full power transmission under MC-NOMA. Its WGPTM grows with subchannels since the co-channel interference decreases  in a subchannel, accelerating the transmissions and extending the local training time. 

By comparing \textbf{Algorithm \ref{Algorithm 2}} to its reduced version of optimal power allocation in both Figs.~\ref{fig:dogl}(a) and~\ref{fig:dogl}(b), 
the benefit of the joint power and bandwidth allocation is demonstrated when a balanced (or reasonably even) distribution of the users among the subchannels is difficult to achieve, e.g., when $N=10$ and $K=25$ in Fig.~\ref{fig:dogl}(b). Specifically, when $N=10$ and $K=25$, half of the subchannels accommodate two users and the other half accommodate three users. The subchannels accommodating three users suffer from stronger interference. This is a particularly unbalanced situation with the largest number of subchannels undergoing stronger interference than the rest, compared to the other $N$ values. Without adjusting the bandwidth of the subchannels, the model uploading times are lengthened, the model training times are shortened, and the data utilized in the training is reduced in the subchannels with three users under the power-only strategy. This leads to more noticeable degradation of the WGPTM under the power-only strategy with respect to the joint strategy, compared to the other $N$ values.

Fig.~\ref{fig:dogl}(c) shows that the WGPTM grows linearly with the duration of a WFL round under all considered algorithms, since the increasing round duration extends the local training times of the users. The figure also shows that the proposed joint power and bandwidth allocation of MC-NOMA under Flexible Aggregation, i.e., \textbf{Algorithm 1}, performs consistently the best. 
The reason is that the objective function is 
$\mathop {\max }\limits_{\mathbf{p}^t_{n},B^t_{n},\,\forall n} \; \Phi^t =\sum\limits _{k\in\mathcal{K}}e_{k}\Phi_{k}^t$ in Problem \textbf{P1} and 
$  \mathop {\max }\limits_{\mathbf{p}_{n},B_{n},\forall n} \, \mathop {\min }\limits_{k\in \mathcal{K}}\ \Phi_k$ in Problem {\bf{P6}}. Since $\sum\limits _{k\in\mathcal{K}}e_{k}\Phi_{k}^t \geq \underset{k\in \mathcal{K}}{\min}\, \Phi_k $ always holds, it readily follows that the maximum value of the objective function of {\bf{P1}} is no smaller than that of {\bf{P6}}. This justifies the superiority of the proposed MC-NOMA-based algorithm to its MC-OMA-based counterpart.

\begin{figure}
	\centering
	\subfigure[$N=10$ and $T=10$ s.]
	{
		\begin{minipage}[t]{0.31\textwidth}
			\centering
			\includegraphics[width=5.3cm]{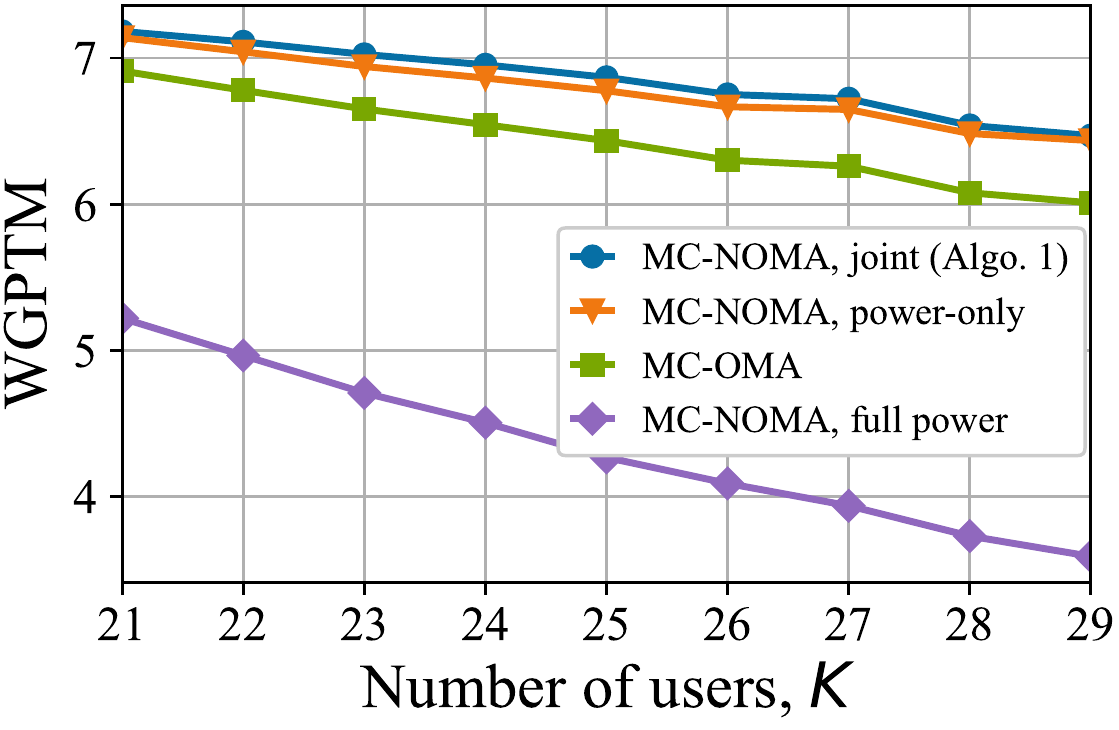}
		\end{minipage}}
	\subfigure[$K=25$ and $T=10$ s.]
	{
		\begin{minipage}[t]{0.31\textwidth}
			\centering
			\includegraphics[width=5.3cm]{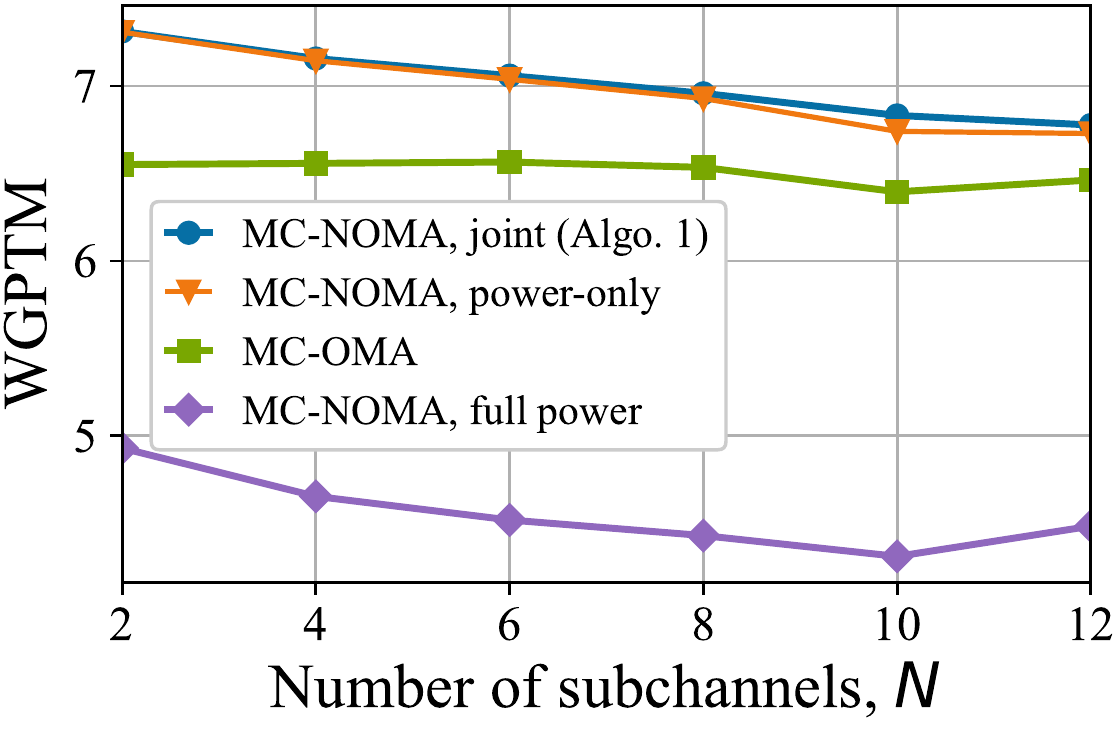}
		\end{minipage}
	}
	\subfigure[$K=25$ and $N=10$.]
	{
		\begin{minipage}[t]{0.31\textwidth}
			\centering
			\includegraphics[width=5.4cm]{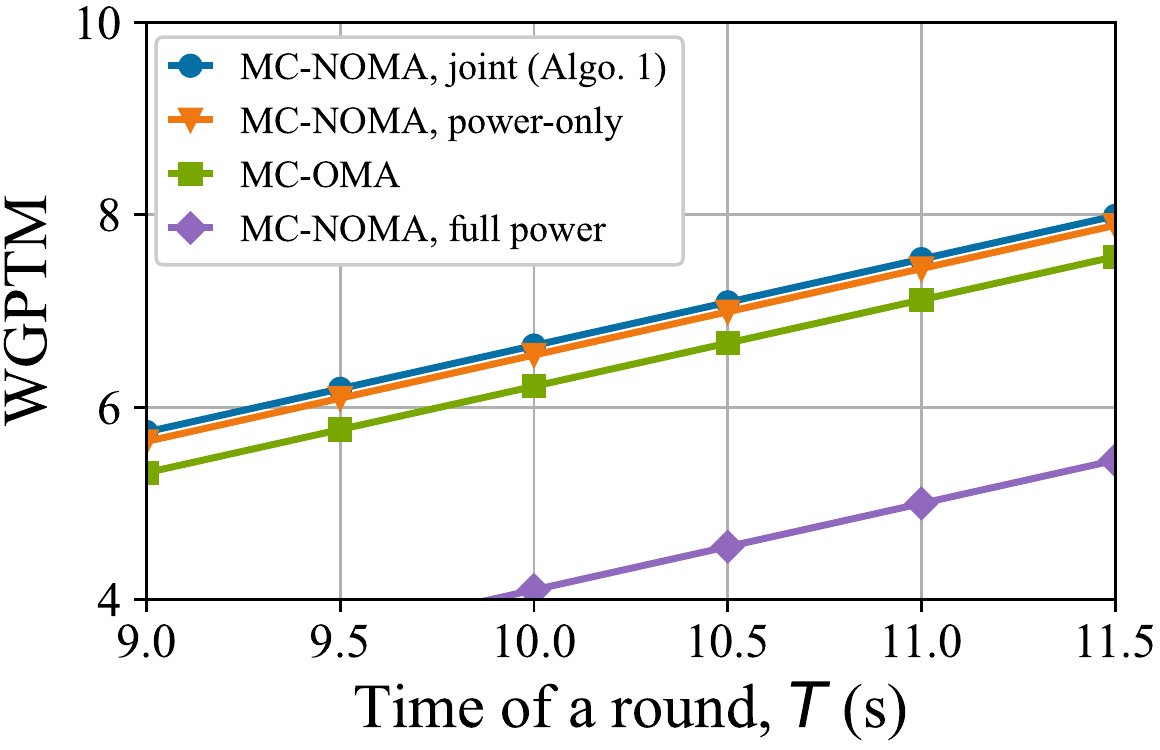}
		\end{minipage}
	}\caption{The WGPTM of the MC-NOMA WFL versus (a) the number of users and (b) the number of subchannels, and (c) the duration of a global round, where random NOMA user clustering is adopted. The CNN model and F-MNIST dataset are considered.
	}\label{fig:random cluster}	
\end{figure}

Fig.~\ref{fig:random cluster} evaluates the scalability of the considered algorithms when the users are clustered at random. The rest of the simulation settings are consistent with Fig.~\ref{fig:dogl}.
While the WGPTMs of the considered algorithms generally decline because of the less effective, random NOMA user clustering, observations  made in Fig.~\ref{fig:random cluster} are consistent with those in Fig.~\ref{fig:dogl}. 
The only exception is the full power transmission under MC-NOMA. As discussed earlier, the users are clustered with a balanced consideration of the number of users per subchannel and the distribution of their channel gains in Fig.~\ref{fig:dogl}. With the increase of subchannels, the impact of the reduced co-channel interference increasingly outgrows that of the reduced bandwidth of each subchannel, leading to the growth of the WGPTM. However, no consideration is given to balancing the distribution of the users' channel gains among different subchannels when the users are clustered randomly. The impact of the reduced bandwidth per subchannel is strong. 

\begin{figure}
	\centering
	\subfigure
	{
		\begin{minipage}[t]{0.46\textwidth}
			\centering
			\includegraphics[width=8cm]{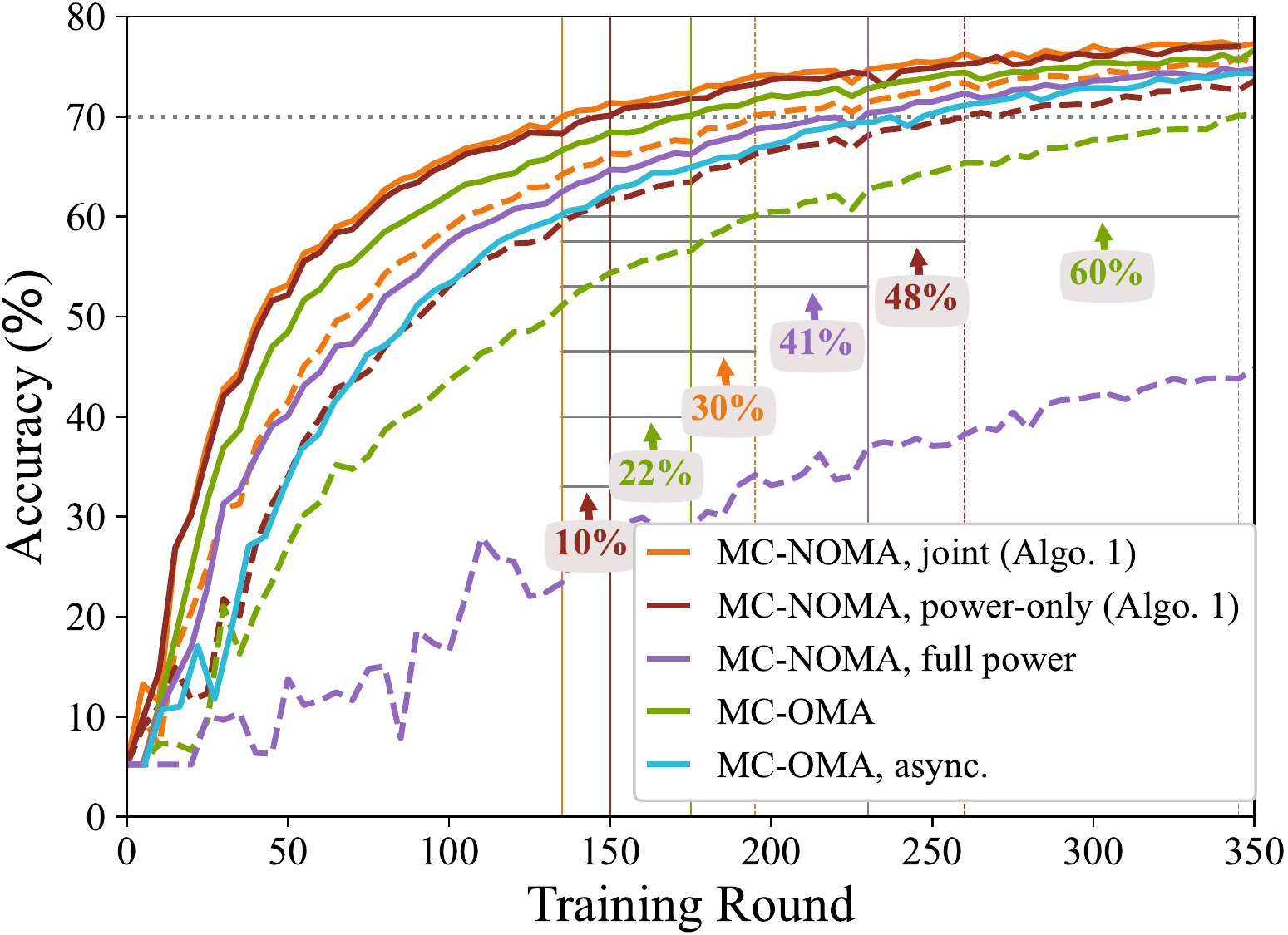}
		\end{minipage}
	}
	\subfigure
	{
		\begin{minipage}[t]{0.46\textwidth}
			\centering
			\includegraphics[width=8cm]{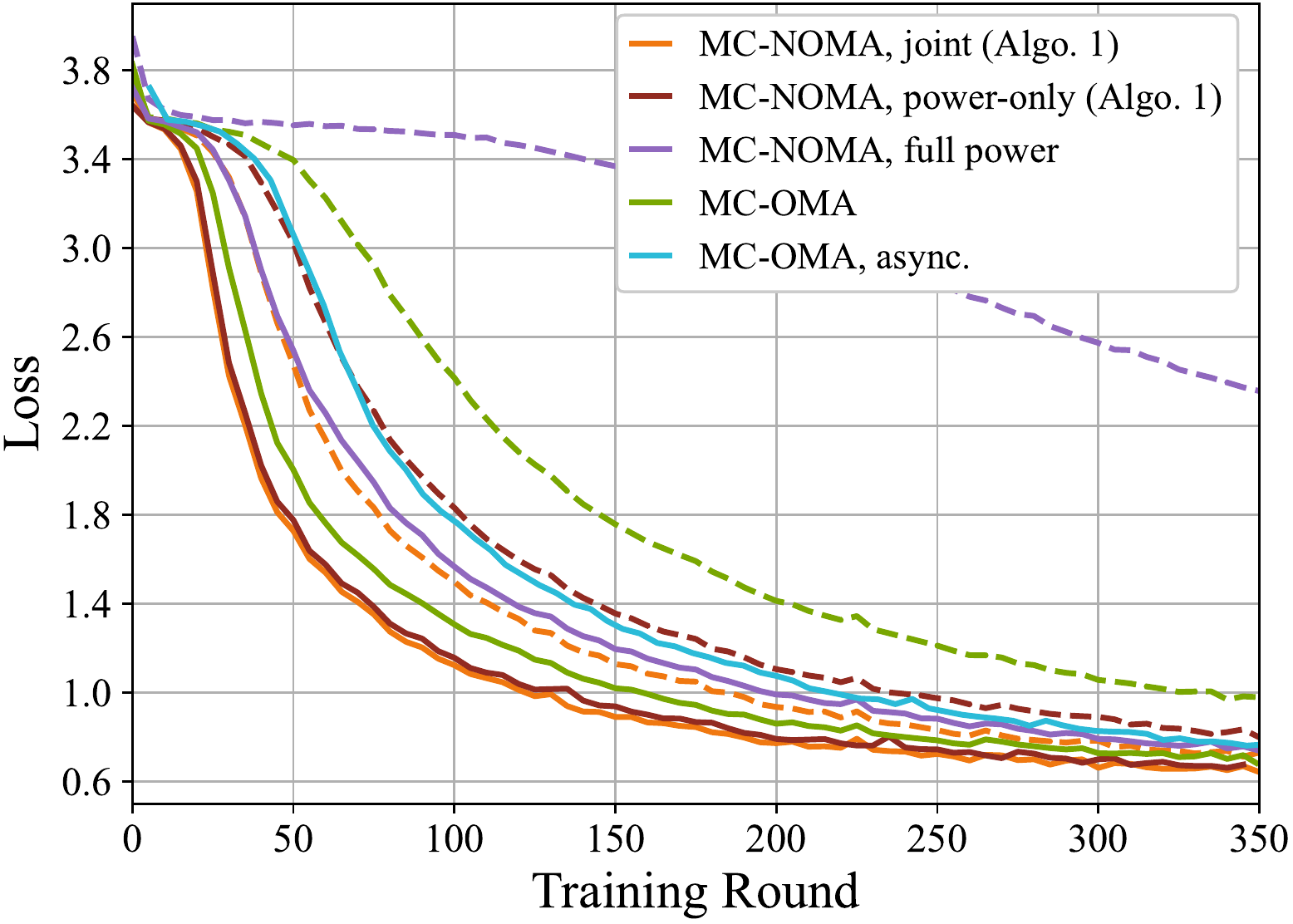}
		\end{minipage}
	}\caption{Training accuracy and loss of MC-NOMA WFL under Flexible Aggregation and Sync-FL, where $K=25$, $T=10\ {\text{s}}$, $N=10$, $\eta=0.03$. CNN and F-MNIST are considered. MC-OMA is plotted under Flexible Aggregation, Sync-FL, and Async-FL. The solid and dotted lines correspond to Flexible Aggregation and Sync-FL, respectively. } 
	\label{fig:acc-loss-FLTA-Supp}
\end{figure}
Fig. \ref{fig:acc-loss-FLTA-Supp} shows the training accuracy and loss of the CNN model trained on the F-MNIST dataset under the Flexible Aggregation, Sync-FL, and Async-FL.
The number of participating users is $25$ in every WFL round. 
We see that the training accuracy increases and the training loss decrease under all the considered power and bandwidth allocation schemes, as the number of rounds increases. 
We also see that the proposed joint optimization of MC-NOMA under Flexible Aggregation i.e., \textbf{Algorithm \ref{Algorithm 2}}, converges the fastest to the highest training accuracy and lowest training loss, demonstrating the effectiveness of the algorithm. 
When the accuracy requirement is 0.7, \textbf{Algorithm \ref{Algorithm 2}} only needs 135 training rounds to reach the requirement, 10\% -- 41\% faster than the other algorithms under Flexible Aggregation, and more than 30\%  faster than the algorithms under Sync-FL. Particularly, \textbf{Algorithm \ref{Algorithm 2}} is 60\% and 41\% faster than the MC-OMA under Sync-FL and Async-FL, respectively. The MC-OMA schemes under Sync-FL~\cite{9264742} and Async-FL~\cite{TT-fed} are the known existing alternatives to \textbf{Algorithm~\ref{Algorithm 2}}.

As shown in Fig.~\ref{fig:acc-loss-FLTA-Supp}, the power allocation plays a more important role than the bandwidth allocation in \textbf{Algorithm \ref{Algorithm 2}}, since the gaps of both the accuracy and loss are marginal between the joint power and bandwidth allocation and the power only allocation. 
Moreover, MC-OMA/TDMA is considerably better than MC-NOMA with full power transmission. This is due to the excessive co-channel interference arising from the use of NOMA when adequate power allocation is absent. 
As also shown in Fig.~\ref{fig:acc-loss-FLTA-Supp}, the Flexible Aggregation
consistently outperforms Sync-FL.
Specifically, the Flexible Aggregation can make a significant difference in the case of the MC-NOMA with full power transmissions. This is because strong co-channel interference results in excessively long transmit delay and spares little time for local model training, as observed in Figs.~\ref{fig:dogl} and~\ref{fig:random cluster}. 

\begin{figure}
	\centering
	\subfigure[Non-i.i.d. dataset: Non-IID-2.]
	{
		\begin{minipage}[t]{0.45\textwidth}
			\centering
			\includegraphics[width=7.8cm]{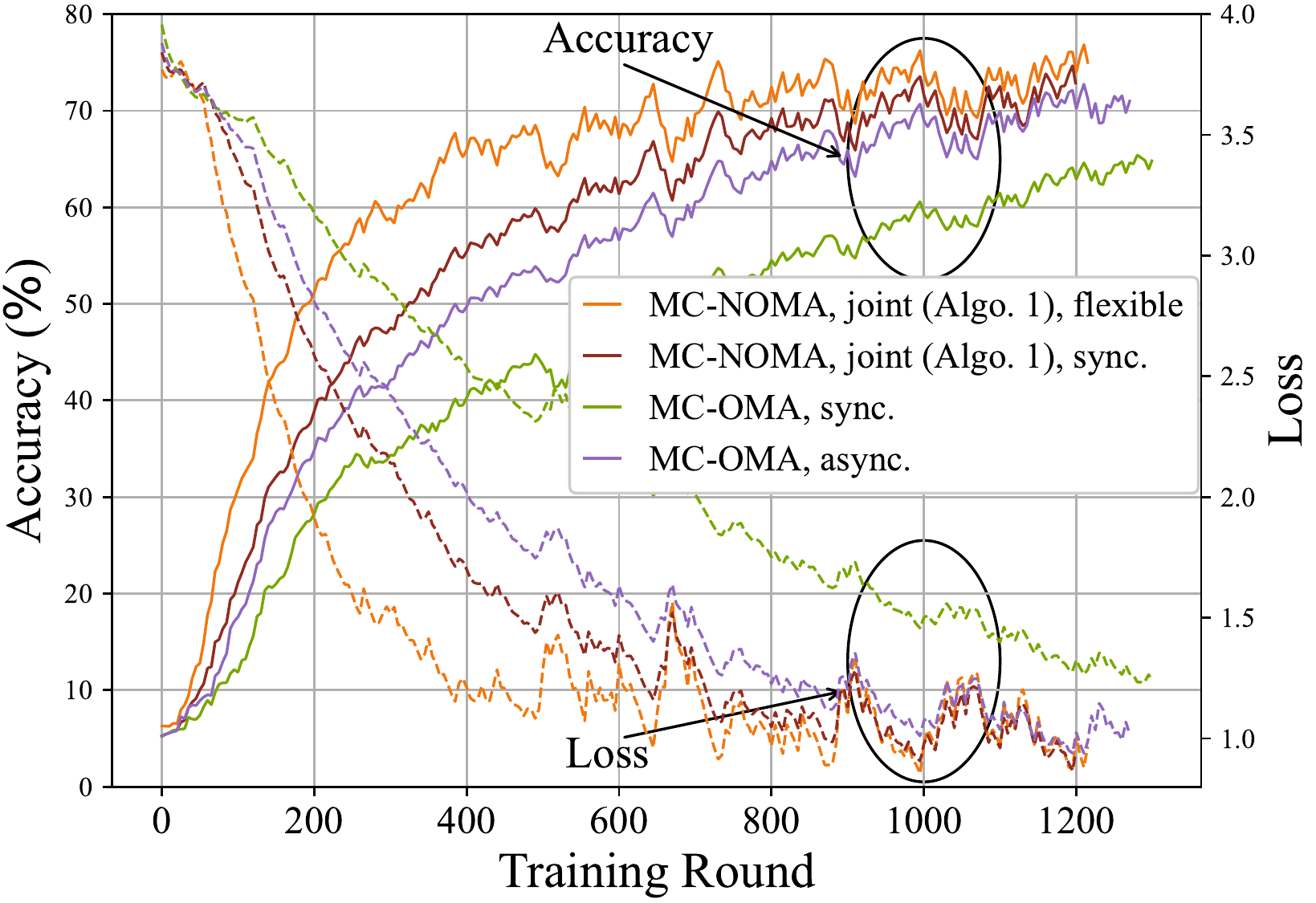}
		\end{minipage}
	}
	\subfigure[Non-i.i.d. datasets: Non-IID-4.]
	{
		\begin{minipage}[t]{0.45\textwidth}
			\centering
			\includegraphics[width=7.8cm]{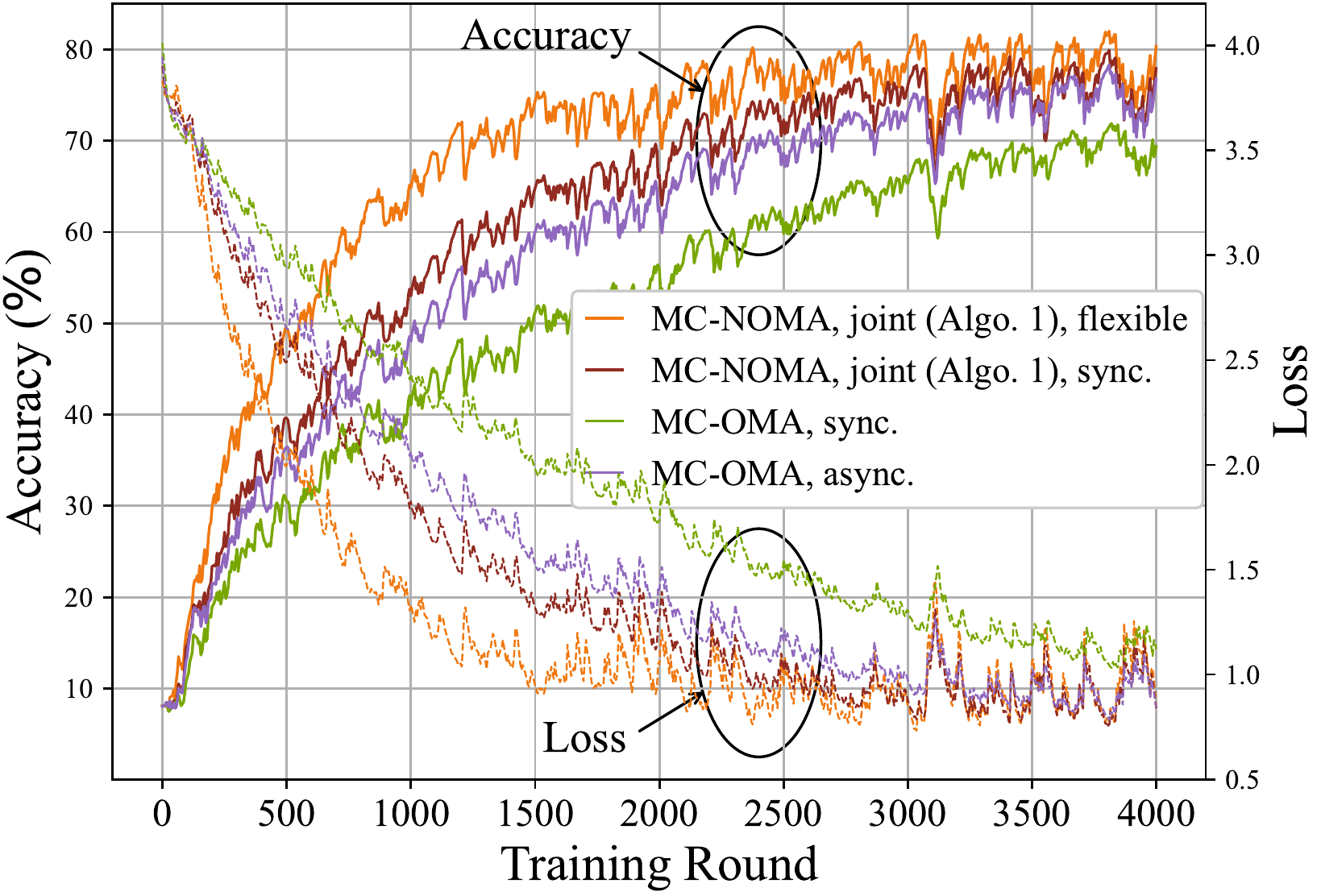}
		\end{minipage}
	}\caption{Training accuracy and loss of MC-NOMA WFL under Flexible Aggregation and Sync-FL, and MC-OMA WFL under Sync-FL and Async-FL. $K=25$, $T=10\ {\text{s}}$, $N=10$, $\eta=0.03$. CNN and F-MNIST are considered. 
	}
	\label{fig:NONIID}
\end{figure}

\begin{figure}
	\centering
	\subfigure[$N=10$ and $T=30$ s.]
	{
		\begin{minipage}[t]{0.3\textwidth}
			\centering
			\includegraphics[width=5.3cm]{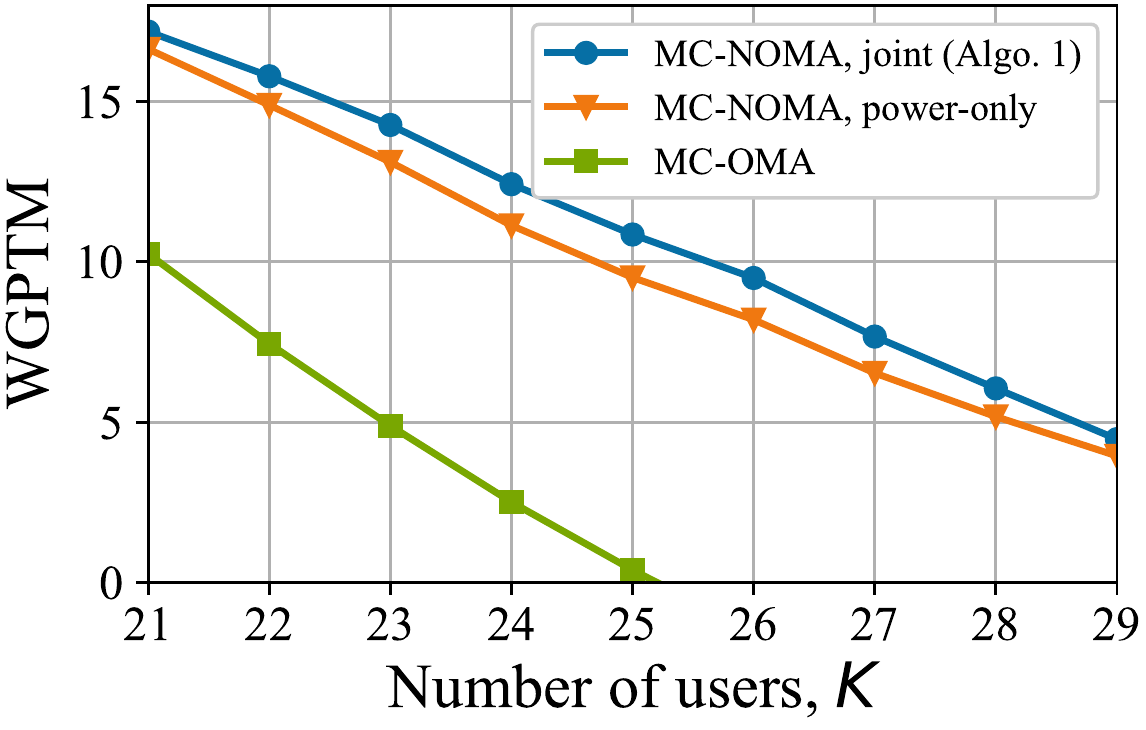}
		\end{minipage}}
	\subfigure[$K=25$ and $T=30$ s.]
	{
		\begin{minipage}[t]{0.3\textwidth}
			\centering
			\includegraphics[width=5.3cm]{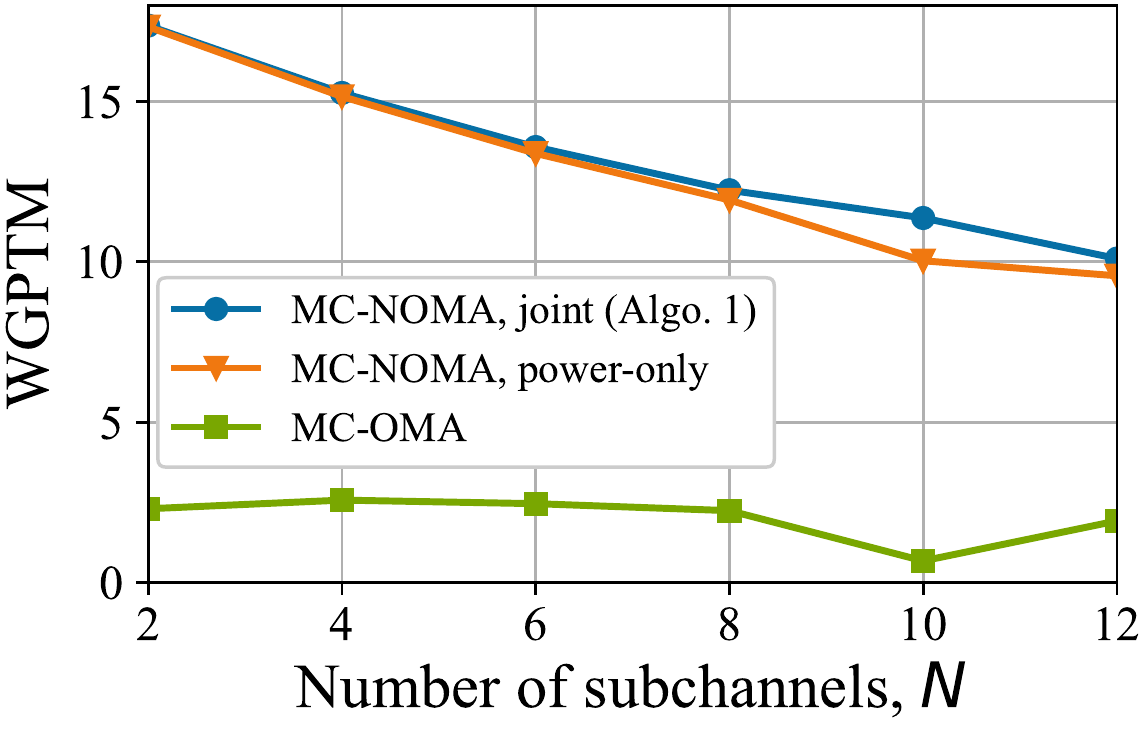}
		\end{minipage}}
	\subfigure[ $K=25$ and $N=10$.]
	{
		\begin{minipage}[t]{0.3\textwidth}
			\centering
			\includegraphics[width=5.3cm]{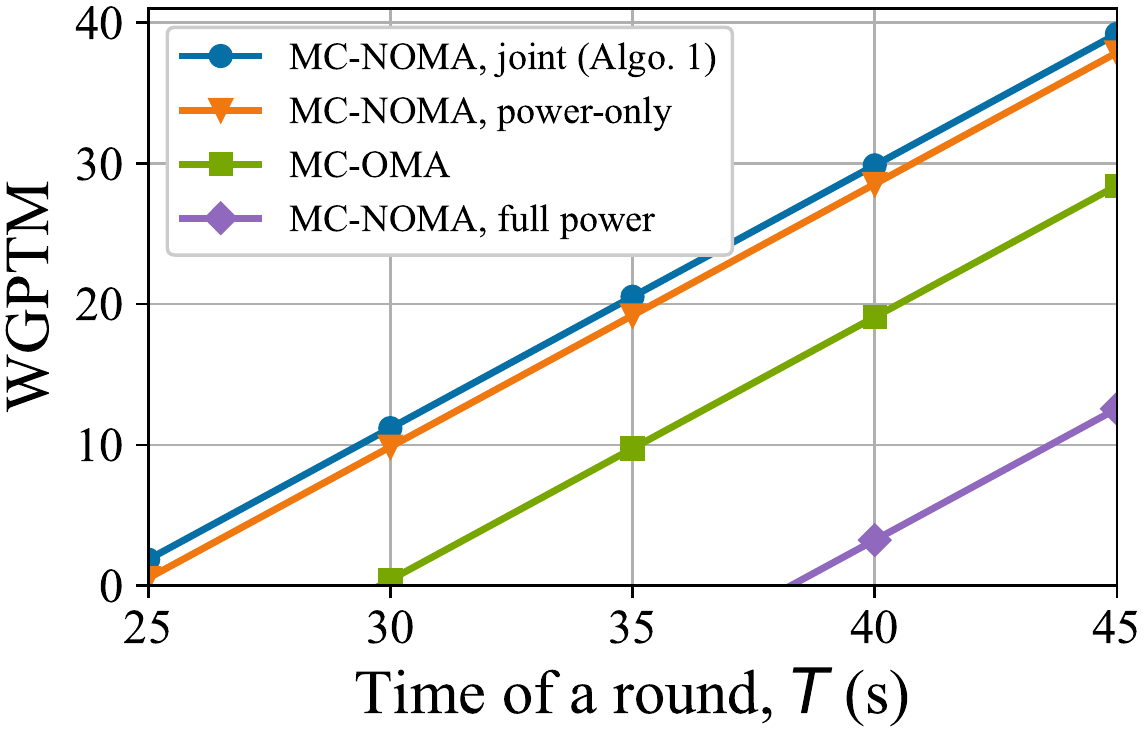}
		\end{minipage}
	}\caption{The WGPTM of MC-NOMA WFL versus (a) the number of users, (b) the number of subchannels, and (c) round duration, where the user clustering method is developed in~\cite{7557079}. ResNet18 and F-CIFAR100 are considered.}
	\label{fig:cifar-wgptm}
\end{figure}	
Fig.~\ref{fig:NONIID} assesses the impact of the non-i.i.d. distribution of the F-MNIST data on the CNN model trained using the proposed \textbf{Algorithm \ref{Algorithm 2}}, where MC-NOMA under both the Flexible Aggregation and Sync-FL are considered. 
For comparison, MC-OMA under Sync-FL and Async-FL are also plotted.
Non-IID-2 and Non-IID-4 are considered in Figs.~\ref{fig:NONIID}(a) and~\ref{fig:NONIID}(b), respectively. 
By comparing Figs.~\ref{fig:acc-loss-FLTA-Supp},~\ref{fig:NONIID}(a), and~\ref{fig:NONIID}(b), we see that the convergence of the algorithm decreases with the increasing difference of the local training datasets.
We also see that MC-NOMA under Flexible Aggregation can still substantially outperform MC-OMA under both Sync-FL and Async-FL, and it takes significantly longer for the algorithms to train on models non-i.i.d. datasets.

\begin{figure}
	\centering
	\subfigure
	{
		\begin{minipage}[t]{0.45\textwidth}
			\centering
			\includegraphics[width=8cm]{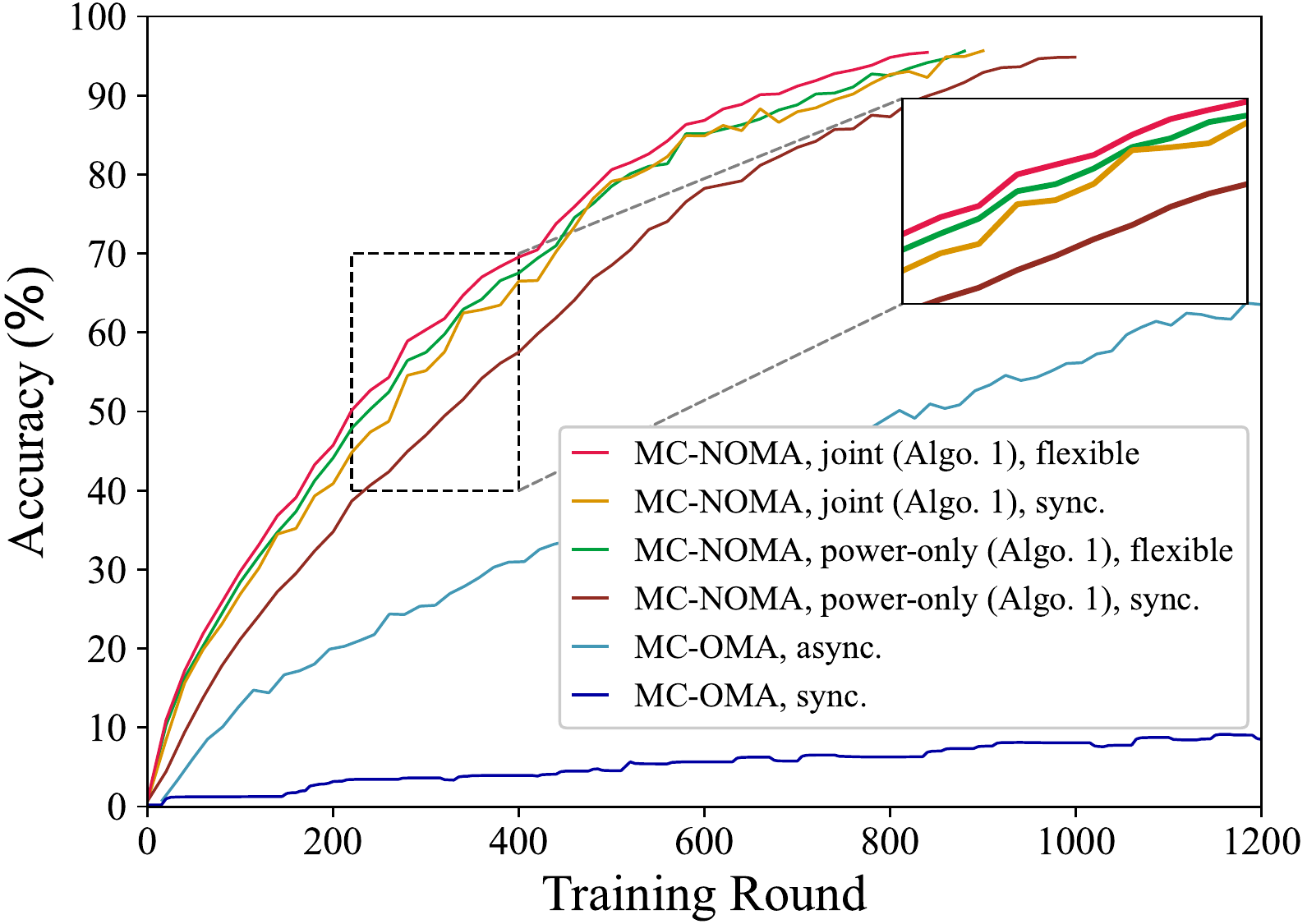}
		\end{minipage}
	}
	\subfigure
	{
		\begin{minipage}[t]{0.45\textwidth}
			\centering
			\includegraphics[width=7.8cm]{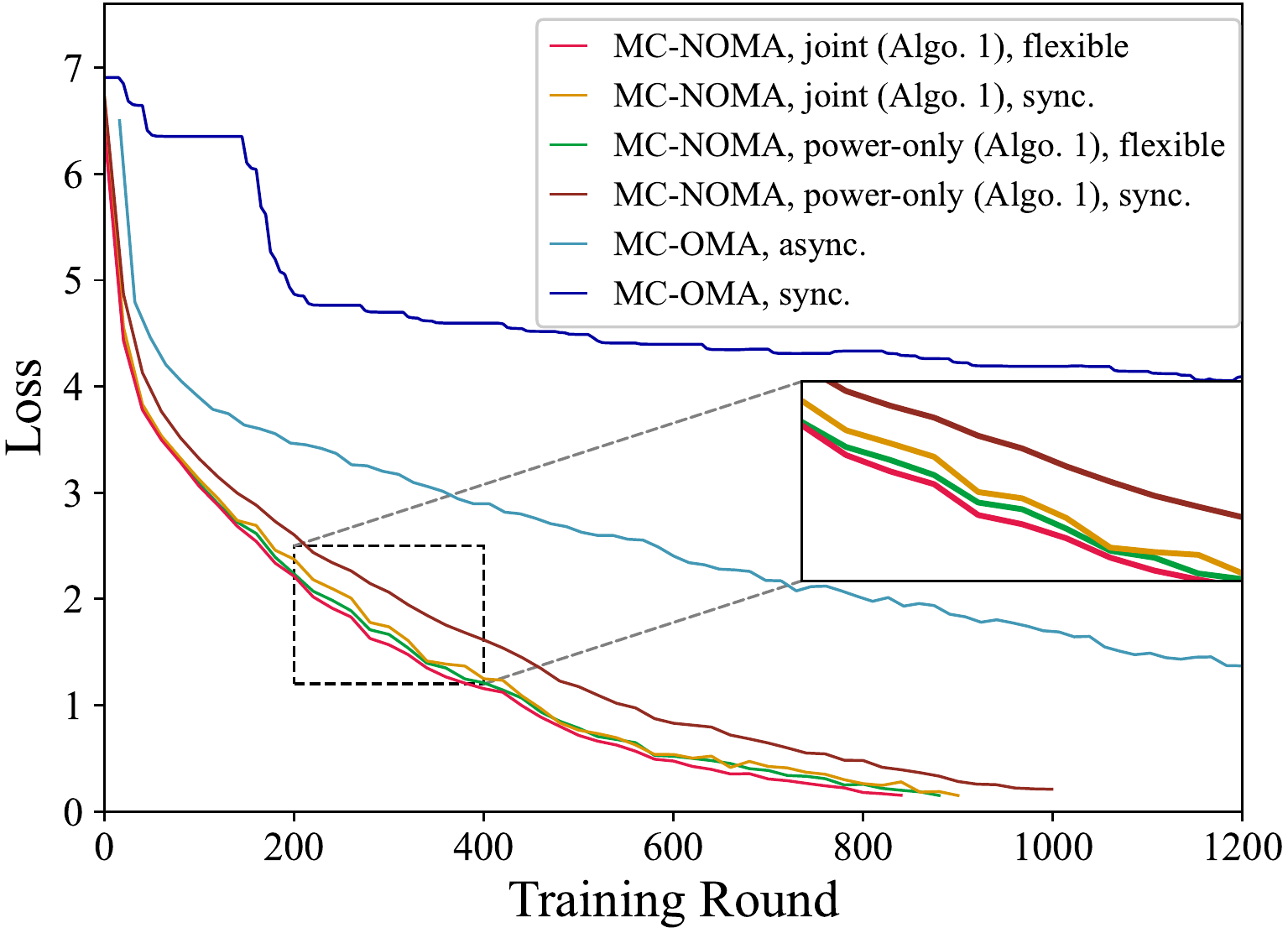}
		\end{minipage}
	}\caption{Training accuracy and loss of MC-NOMA under Flexible Aggregation and Sync-FL, and MC-OMA under Sync-FL and Async-FL.
	 $K=25$, $T=30\ {\text{s}}$, $N=10$, and $\eta=1$. ResNet18 and F-CIFAR100 are considered.  
	}
	\label{fig:Fecifar100}
\end{figure}
Next, we proceed to assess the applicability of the proposed \textbf{Algorithm \ref{Algorithm 2}} to the ResNet18 model trained on the F-CIFAR100 dataset. 
Fig.~\ref{fig:cifar-wgptm} shows the changes in the WGPTMs of the considered algorithms as users, subchannels, and the round duration increase, where the ResNet18 model is trained on the F-CIFAR100 dataset. Observations made in this figure are consistent with those in Fig.~\ref{fig:dogl}. Specifically, the proposed optimal  joint allocation of MC-NOMA under Flexible Aggregation, i.e., \textbf{Algorithm \ref{Algorithm 2}}, outperforms all the other considered algorithms in terms of WGPTM. 
Note that at $T=10$ s, the MC-NOMA with full power transmission is infeasible with a negative WGPTM, as shown in Fig.~\ref{fig:cifar-wgptm}(c), and therefore is not plotted in Figs.~\ref{fig:cifar-wgptm}(a) and~\ref{fig:cifar-wgptm}(b).

Last but not least, Fig. \ref{fig:Fecifar100} demonstrates the accuracy and loss of the ResNet18 model, as the number of training rounds increases.
We plot the MC-NOMA under both Flexible Aggregation and Sync-FL. 
For comparison, we also plot the MC-OMA under Sync-FL and Async-FL, which are the possible existing alternatives to the MC-NOMA under Flexible Aggregation.
We see that the proposed \textbf{Algorithm \ref{Algorithm 2}} leads to the faster and better convergence of the ResNet18 model trained to classify the F-CIFAR100 dataset, compared to its alternative approaches.
The convergence of the ResNet18 model is much slower than that of the CNN model, due to the substantially bigger size and more complex structure of the ResNet18 model and the larger color images of the F-CIFAR100 dataset. 
As a consequence, the duration of $T=30$ seconds per WFL round may not be able to sustain the effective model uploading of some users when their channel conditions are poor in some WFL rounds. In the case of Sync-FL, no local training takes place in those rounds. Consequently, the training accuracy is poor and the convergence is slow. Fig.~\ref{fig:Fecifar100} validates the applicability of the proposed MC-NOMA under Flexible Aggregation and the optimal joint allocation to the ResNet18 model and F-CIFAR100 dataset.

	\section{Conclusion}\label{section:con}
In this paper, we proposed an MC-NOMA WFL system that allows users to train different numbers of iterations per WFL round, adapting to their channel conditions and computing powers. 
A new metric, WGPTM, was designed to measure the convergence of the new system based on a convergence analysis. 
A new, non-convex, joint power and bandwidth allocation problem was formulated to maximize the WGPTM and, in turn, harness the convergence of the new system.
Using variable substitution and Cauchy's inequality, we converted the problem losslessly to a convex problem solved with polynomial complexity.
Extensive simulations based on the F-MNIST and F-CIFAR100 datasets showed that the new MC-NOMA WFL can efficiently reduce communication delay, increase local model training times, and accelerate convergence.

	\appendix
	\subsection{Proof of Theorem 1}\label{proof-theorem}
	Under the assumptions \eqref{ass1} and \eqref{ass2}, we have
		\begin{align}
			F_{k}\left(\mathbf{w}_k^t\right)-F_{k}^{*}\leq\left(1-c_k\right)^{\phi_k}\left(F_{k}\left(\mathbf{w}_{\mathcal{G}}^{t-1}\right)-F_{k}^{*}\right),\label{lemma3}
		\end{align}
		where $c_k=2m\eta_k-mL\eta_k^{2}$
by combining 
	\cite[Eq.~9.9]{boyd2004convex} and \cite[Eq.~9.17]{boyd2004convex}.
	We also have 
{\small			\begin{subequations}
			\begin{align}
				\mathbb{E}\big(\underset{k\in\mathcal{K}}{\sum}e_{k}\left\Vert \mathbf{w}_{\mathcal{G}}^{t}-\mathbf{w}_{k}^{t}\right\Vert \big)
				&=\mathbb{E}\big(\underset{k\in\mathcal{K}}{\sum}e_{k}\left\Vert \mathbf{w}_{\mathcal{G}}^{t-1}-\mathbf{w}_{k}^{t}+\mathbf{w}_{\mathcal{G}}^{t}-\mathbf{w}_{\mathcal{G}}^{t-1}\right\Vert\big)
				\leq\mathbb{E}\big(\underset{k\in\mathcal{K}}{\sum}e_{k}\left\Vert \mathbf{w}_{\mathcal{G}}^{t-1}-\mathbf{w}_{k}^{t}\right\Vert \big)\label{41a}
				\\
				&=\mathbb{E}\big(\underset{k\in\mathcal{K}}{\sum}e_{k}\Big\Vert \stackrel[c=1]{\phi_{k}^t}{\sum}\left(\eta_k\nabla F_{k}\left(\mathbf{w}_{k,c}^{t-1}\right)\right)\Big\Vert \big)\label{41b}
				\\
				&\leq\mathbb{E}\big(\underset{k\in\mathcal{K}}{\sum}e_{k}\stackrel[c=1]{\phi_{k}^t}{\sum}\eta_k\Big\Vert \nabla F_{k}\left(\mathbf{w}_{k,c}^{t-1}\right)\Big\Vert \big) 
				\leq \underset{k\in\mathcal{K}}{\sum}e_{k}\stackrel[c=1]{\phi_{k}^t}{\sum}\eta_kG\label{41c}
				\\
				\notag
				&= \underset{k\in\mathcal{K}}{\sum}e_{k}\phi_{k}^t\eta_kG =\underset{k\in\mathcal{K}}{\sum}e_{k}\Phi_{k}^t\eta G, 
			\end{align}
		\end{subequations}}where \eqref{41a} is due to the fact that $\mathbb{E}\left\Vert X-\mathbb{E}X\right\Vert \leq\mathbb{E}\left\Vert X\right\Vert $ and we set $X=\mathbf{w}_{\mathcal{G}}^{t-1}-\mathbf{w}_{k}^{t}$ with probability $e_k$; \eqref{41b} plugs in the $c$-th local training of user $k$ based on its mini-batches, i.e., $\mathbf {w}_{k,c+1}^t=\mathbf {w}_{k,c}^t-\eta_k\nabla F_k(\mathbf {w}_{k,c}^t)$, $c=1,\cdots,\phi_{k}^t$; and \eqref{41c} is due to the assumption that $\mathbb{E}\{\left\Vert \nabla F_{k}\right\Vert ^{2}\}\leq G^{2}$ in \eqref{ass3}.
		As a result, the distance between $\mathbf{w}_{k}^{t}$ and $\mathbf{w}_{\mathcal{G}}^{t}$ is upper bounded by
		\begin{align}\label{eq: lemma 4} 
			\mathbb{E}\big(\underset{k\in\mathcal{K}}{\sum}e_{k}\left\Vert \mathbf{w}_{\mathcal{G}}^{t}-\mathbf{w}_{k}^{t}\right\Vert \big)\leq\underset{k\in\mathcal{K}}{\sum}e_{k} \Phi_{k}^t\eta G.
		\end{align}
		
        Further, after the $t$-th model aggregation, the divergence of the global loss function value $F(\mathbf{w}_{\mathcal{G}}^{t})$ and the local loss function values $\{F_{k}\left(\mathbf{w}_{k}^{t}\right),\forall k\}$ is given by
{\small		\begin{subequations}\label{ineq:44}
			\begin{align}
				\mathbb{E}\big(\underset{k\in\mathcal{K}}{\sum}e_{k}\left(F_{k}\left(\mathbf{w}_{\mathcal{G}}^{t}\right)-F_{k}\left(\mathbf{w}_{k}^{t}\right)\right)\big)
				&\leq\mathbb{E}\big(\underset{k\in\mathcal{K}}{\sum}e_{k}\left(\left\langle\nabla F_{k}\left(\mathbf{w}_{\mathcal{G}}^{t}\right),\mathbf{w}_{\mathcal{G}}^{t}-\mathbf{w}_{k}^{t}\right\rangle\right)\big)\label{42a}\\
				&\leq\mathbb{E}\big(\underset{k\in\mathcal{K}}{\sum}e_{k}\big(\left\Vert \nabla F_{k}\left(\mathbf{w}^{t}\right) \right\Vert \left\Vert \mathbf{w}_{\mathcal{G}}^{t}-\mathbf{w}_{k}^{t}\right\Vert\big)\big)\label{42b}\\
				&\leq\underset{k\in\mathcal{K}}{\sum}e_{k}\Phi_{k}^t\eta G^{2},\label{42c}
			\end{align}
		\end{subequations}}where \eqref{42a} is based on the equivalent form of convex assumption in \eqref{ass1}, $f(\mathbf{y})-f(\mathbf{x})\leq\left\langle \nabla f(\mathbf{y}),\mathbf{y}-\mathbf{x}\right\rangle $;
		\eqref{42b} is based on the property of vector inner product $\left\langle {\mathbf{a},\mathbf{b}}\right\rangle \leq\left\Vert \mathbf{a}\right\Vert \cdot \left\Vert \mathbf{b}\right\Vert$, and \eqref{42c} is based on \eqref{eq: lemma 4} and $\mathbb{E}\{\left\Vert \nabla F_{k}\right\Vert ^{2}\}\leq G^{2}$ in \eqref{ass3}.

		Now, we can expand $\mathbb{E}\big(\underset{k\in\mathcal{K}}{\sum}e_{k}\left(F_{k}\left(\mathbf{w}_{\mathcal{G}}^{t}\right)-F_{k}^{*}\right)\big)$ as
		{\small \begin{align}
			\mathbb{E}\Big(\underset{k\in\mathcal{K}}{\sum}e_{k}&\left(F_{k}\left(\mathbf{w}_{\mathcal{G}}^{t}\right)-F_{k}^{*}\right)\Big) =\mathbb{E}\Big(\underset{k\in\mathcal{K}}{\sum}e_{k}\big(F_{k}\left(\mathbf{w}_{\mathcal{G}}^{t}\right)-F_{k}\left(\mathbf{w}_{k}^{t}\right)\big)\Big)+
			\mathbb{E}(\underset{k\in\mathcal{K}}{\sum}e_{k}\left(F\left(\mathbf{w}_{k}^{t}\right)-F_{k}^{*}\right)\Big)\label{eq:45}.
		\end{align}}
		Based on (\ref{lemma3}), we have
		\begin{align}
			\underset{k\in\mathcal{K}}{\sum}e_{k}\left(F\left(\mathbf{w}_k^t\right)-F_{k}^{*}\right)\leq\underset{k\in\mathcal{K}}{\sum}e_{k}\left(1-c_k\right)^{\Phi_{k}^t}\left(F_{k}\left(\mathbf{w}_{\mathcal{G}}^{t-1}\right)-F_{k}^{*}\right).\label{lemma3-2}
		\end{align}
		By combining \eqref{ineq:44}, \eqref{eq:45}, and (\ref{lemma3-2}), we have
{\small	\begin{align}
			 \mathbb{E}	\Big(\underset{k\in\mathcal{K}}{\sum}e_{k}\big(F_{k}&\left(\mathbf{w}_{\mathcal{G}}^{t}\right)-F_{k}^{*}\big)\Big)\leq\mathbb{E}\Big(\underset{k\in\mathcal{K}}{\sum}e_{k}\left(1-c_k\right)^{\Phi_{k}^t}\times 
            \big(F_{k}\left(\mathbf{w}_{\mathcal{G}}^{t-1}\right)	
			-F_{k}^{*}\big)\Big)+\underset{k\in\mathcal{K}}{\sum}e_{k}\Phi_{k}^t\eta G^{2}.\label{ieq1}
		\end{align}}
		Subtracting $\mathbb{E}\big(\underset{k\in\mathcal{K}}{\sum}e_{k}\left(F_{k}\left(\mathbf{w}^{t-1}\right)-F_{k}^{*}\right)\big)$ from both sides of (\ref{ieq1}), we have
        {  \small
		\begin{subequations}
			\begin{align}
				\notag \mathbb{E}\Big(F&\left(\mathbf{w}_{\mathcal{G}}^{t}\right)-F\left(\mathbf{w}_{\mathcal{G}}^{t-1}\right)\Big)	=\mathbb{E}\Big(\underset{k\in\mathcal{K}}{\sum}e_{k}\left(F_{k}\left(\mathbf{w}_{\mathcal{G}}^{t}\right)-F_{k}^{*}\right)\Big)-\mathbb{E}\Big(\underset{k\in\mathcal{K}}{\sum}e_{k}\Big(F_{k}\left(\mathbf{w}_{\mathcal{G}}^{t-1}\right)-F_{k}^{*}\Big)\Big)\\
                \leq&\mathbb{E}\Big(\underset{k\in\mathcal{K}}{\sum}e_{k}\left(\left(1\!\!-\!\!c_k\right)^{\Phi_{k}^t}\!\!-1\right)\left(F_{k}\left(\mathbf{w}_{\mathcal{G}}^{t-1}\right)-F_{k}^{*}\right)\Big)+\underset{k\in\mathcal{K}}{\sum}e_{k}\Phi_{k}^t\eta G^{2}\label{47a}\\
				=&\mathbb{E}\big(\underset{k\in\mathcal{K}}{\sum}e_{k}\left(-\Phi_{k}^tc_k+\mathcal{O}\left(({c_k})^2\right)\right)\left(F_{k}\left(\mathbf{w}_{\mathcal{G}}^{t-1}\right)-F_{k}^{*}\right)\big)+\underset{k\in\mathcal{K}}{\sum}e_{k}\Phi_{k}^t\eta G^{2}\label{47b}	\\
				\leq&-\big(\underset{k\in{\cal K}^t_{n}}{\min}\{c_k\mathbb{E}\left(F_{k}\left(\mathbf{w}_{\mathcal{G}}^{t-1}\right)-F_{k}^{*}\right)\}-\eta G^{2}\big)\underset{k\in\mathcal{K}}{\sum}e_{k}\Phi_{k}^t	+\underset{k\in\mathcal{K}}{\sum}e_{k}\mathcal{O}\left((c_k)^{2}\right)\mathbb{E}\left(F_{k}\left(\mathbf{w}_{\mathcal{G}}^{t-1}\right)-F_{k}^{*}\right),
				\label{47c}	
			\end{align}		
		\end{subequations}}where 
		\eqref{47b} is obtained by expanding $\left(1-c_k\right)^{\Phi_{k}^t}-1$, and \eqref{47c} is due to deformation of \eqref{47b}.

	\subsection{Proof of Lemma 2}\label{proof-lemma2}
	The determinant of the Hessian matrix  of $w_1^n$, i.e., $\det\left(\nabla^{2}w_{1}^{n}\right)$, is given by
    {\small
		\begin{align}
			\det\left(\nabla^{2}w_{1}^{n}\right)	=&\frac{3\left|h_{n,|\mathcal{K}_{n}|}\right|^{2}p_{n,|\mathcal{K}_{n}|}^{*}-2^{A_{|\mathcal{K}_{n}|-1}(n)+1}\ln\big(\frac{\left|h_{n,|\mathcal{K}_{n}|}\right|^{2}p_{n,|\mathcal{K}_{n}|}^{*}}{2^{A_{|\mathcal{K}_{n}|-1}(n)}}+1\big)}{\big(\left|h_{n,|\mathcal{K}_{n}|}\right|^{2}p_{n,|\mathcal{K}_{n}|}^{*}+2^{A_{|\mathcal{K}_{n}|-1}(n)}\big)^{2}}			\times\frac{\left(S\beta_{\mathcal{K}_{n}(|\mathcal{K}_{n}|)}\right)^{2}\left|h_{n,|\mathcal{K}_{n}|}\right|^{2}p_{n,|\mathcal{K}_{n}|}^{*}}{\big(B_{n}\log_{2}\big(\frac{\left|h_{n,|\mathcal{K}_{n}|}\right|^{2}p_{n,|\mathcal{K}_{n}|}^{*}}{2^{A_{|\mathcal{K}_{n}|-1}(n)}}+1\big)\big)^{4}}.
			\label{det-hessian}
		\end{align}	}
According to the inequality $x\ln (1+\frac{A}{x})\leq A,\forall x>0$, we obtain
	\begin{align}	
3\left|h_{n,|\mathcal{K}_{n}|}\right|^{2}p_{n,|\mathcal{K}_{n}|}^{*}-2^{A_{|\mathcal{K}_{n}|-1}(n)+1}\ln\big(\frac{\left|h_{n,|\mathcal{K}_{n}|}\right|^{2}p_{n,|\mathcal{K}_{n}|}^{*}}{2^{A_{|\mathcal{K}_{n}|-1}(n)}}+1\big)  \geq
\left|h_{n,|\mathcal{K}_{n}|}\right|^{2}p_{n,|\mathcal{K}_{n}|}^{*}>0.\label{he1}
		\end{align}
The other parts of \eqref{det-hessian} are all non-negative, i.e.,	
$\big(\left|h_{n,|\mathcal{K}_{n}|}\right|^{2}p_{n,|\mathcal{K}_{n}|}^{*}+2^{A_{|\mathcal{K}_{n}|-1}(n)}\big)^{2}>0$, \\
$\left(S\beta_{\mathcal{K}_{n}(|\mathcal{K}_{n}|)}\right)^{2}\left|h_{n,|\mathcal{K}_{n}|}\right|^{2}p_{n,|\mathcal{K}_{n}|}^{*}>0$, 
and 
$\big(B_{n}\log_{2}\big(\frac{\left|h_{n,|\mathcal{K}_{n}|}\right|^{2}p_{n,|\mathcal{K}_{n}|}^{*}}{2^{A_{|\mathcal{K}_{n}|-1}(n)}}+1\big)\big)^{4}>0$.
	As a result, $\nabla^{2}w_{1}^{n}\succeq0$, is a positive semi-definite matrix and $w_{1}^{n}$ is convex in $A_{|\mathcal{K}_{n}|-1}(n)$ and~$B_{n}$.
	 Likewise, the determinant of the Hessian matrix of $z_2^n$, $\det\left(\nabla^{2}z_{2}^{n}\right)$, is
	\begin{align}		\det(\nabla^{2}z_{2}^{n})=\big(S\big(\sum\limits _{i=1}^{|\mathcal{K}_{n}|-1}\sqrt{\beta_{\mathcal{K}_{n}(i)}}\big)^{2}\big)^{2}\times&\frac{\left(3\left(A_{|\mathcal{K}_{n}|-1}(n)-A_{0}(n)\right)-\frac{2}{\ln2}\right)}{B_{n}^{4}\left(A_{|\mathcal{K}_{n}|-1}(n)-A_{0}(n)\right)^{5}}\geq0,
	\end{align}
since $3\left(A_{|\mathcal{K}_{n}|-1}(n)-A_{0}(n)\right)\geq\frac{2}{\ln2}$ following the inherent condition $|h_{n,1}|^2p_{n,1}>N_0B_n$.
Therefore, the Hessian matrix, $\nabla^{2}z_{2}^{n}\succeq0$, is positive semi-definite and $z_{2}^{n}$ is convex in  $A_{|\mathcal{K}_{n}|-1}(n)$ and $B_{n}$.
	
	\subsection{Optimal Joint Power and Bandwidth Allocation under MC-NOMA Sync-FL}
In the Sync-FL where all users train the same number of iterations per WFL round, we maximize the consistent LPTM of all users by solving a max-min problem:
		\begin{align}
			\text{\bf{P6}}:  &\mathop {\max }\limits_{\mathbf{p}_{n},B_{n},\forall n} \ \min_{k\in \mathcal{K}}\ \Phi_k 
			\quad\mathrm{s.t. } \;
			\text{\eqref{constraints a} --
   \textbf{\eqref{constraints c}}},  \notag
		\end{align}
where $\Phi_{k}=\frac{\phi_{k}}{|{\cal M}_{k}|}=\frac{T_{{\rm{c}},k}{ \beta_k }}{\alpha|{\cal M}_{k}|}=\frac{(T_k-T_{{\rm{d}},k}-T_{{\rm{u}},k}){ \beta_k }}{\alpha|{\cal M}_{k}|}$ based on \eqref{eq:fl delay}--\eqref{iteration}. 
Problem \textbf{P6} can be rewritten as 
	\begin{align}
    \mathop {\min }\limits_{\mathbf{p}_{n},B_{n},\forall n} \max_{k\in \mathcal{K}}   
   \frac{T_{{\rm{u}},k}{ \beta_k }}{\alpha|{\cal M}_{k}|} 
	\quad	\mathrm{s.t.} \;
		\text{ \eqref{constraints a} --
  \textbf{\eqref{constraints b}}},  \notag 
	\end{align}
which can be further decoupled between the $N$ subchannels. 
   The $n$-th subproblem corresponding to the $n$-th subchannel, $\forall n$, is  $\min\limits_{\mathbf{p}_{n}} \max\limits_{ k\in \mathcal{K}_n} \frac{T_{{\rm{u}},k}{ \beta_k }}{\alpha|{\cal M}_{k}|},\, \text{s.t \eqref{constraints b}}$, since \eqref{constraints a} and \eqref{constraints e} only depend on $B_n$.
The solution to the subproblem is only taken when  $\frac{T_{{\rm{u}},k}{ \beta_k }}{\alpha|{\cal M}_{k}|}$ is equal $ \forall k \in \mathcal{K}_n$ and denoted by
\begin{align}
     Q'_n= Q'_n(B_n)\triangleq\frac{T_{{\rm{u}},k}{ \beta_k }}{\alpha|{\cal M}_{k}|},\ \forall k \in \mathcal{K}_n.\label{p7:1}
\end{align}
Then, $ \max\limits_{k\in \mathcal{K}}\ \frac{T_{{\rm{u}},k}{ \beta_k }}{\alpha|{\cal M}_{k}|} = \max\limits_n\ Q'_n(B_n)$. 
Problem \textbf{P7} can be reformulated as
	\begin{align}
		\text{\bf{P8}}:  &\min\limits_{B_n,\forall n} \max\limits_{n}\;  Q'_n(B_n) \;\;\text{ s.t. \eqref{constraints a}, \eqref{constraints e}}.  \notag 
	\end{align}
Since $\frac{T_{{\rm{u}},k}{ \beta_k }}{\alpha|{\cal M}_{k}|}$ is a convex and monotonically decreasing function of $B_n$, $Q'_n(B_n)=\min\limits_{\mathbf{p}_{n}} \max\limits_{k\in \mathcal{K}_n} \frac{T_{{\rm{u}},k}{ \beta_k }}{\alpha|{\cal M}_{k}|}  $ is a convex and decreasing function of $B_n$;
see \cite[Sec. 3.2.3]{boyd2004convex}.
As a result, $\max\limits_{{n}}\;  Q'_n(B_n)$ is a convex function of $B_n,\forall n$. Problem \textbf{P8} has a convex objective and liner constraints. 

By variable substitution using \eqref{eq:APi}, we rewrite \eqref{p7:1} as
\begin{align}
Q'_n &=\frac{S\beta_{\mathcal{K}_{n}(|\mathcal{K}_n|)}/(B_n\alpha|{\cal M}_{\mathcal{K}_{n}(|\mathcal{K}_n|}|)}{A_{|\mathcal{K}_{n}|}({n})-A_{|\mathcal{K}_{n}|-1}({n})}
=\frac{\sum_{i=1}^{|\mathcal{K}_{n}|-1}\left(S\beta_{\mathcal{K}_{n}(i)}/(B_{n}\alpha|{\cal M}_{\mathcal{K}_{n}(i)}|)\right)}{A_{|\mathcal{K}_{n}|-1}({n})-A_{0}({n})}.\label{P7:2}
\end{align}

Since $\max \limits_{k\in \mathcal{K}_n}{\frac{T_{{\rm{u}},k}{ \beta_k }}{\alpha|{\cal M}_{k}|}}$ decreases with the increase of $p_{n,|\mathcal{K}_n|}$, $p_{n,|\mathcal{K}_n|}^*=P_{\max}$ when the minimum of $\max \limits_{k\in \mathcal{K}_n}{\frac{T_{{\rm{u}},k}{ \beta_k }}{\alpha|{\cal M}_{k}|}}$ is taken.
From \eqref{A_K}, $A_{|\mathcal{K}_{n}|}({n})$ can be expressed with to $A_{|\mathcal{K}_{n}|-1}({n})$ and $B_n$. Substituting \eqref{A_K} into \eqref{P7:2} and solving \eqref{P7:2}, we can obtain the optimal $A^*_{|\mathcal{K}_{n}|-1}({n})$ as a function of $B_n$. By substituting $A^*_{|\mathcal{K}_{n}|-1}({n})$ into \eqref{P7:2}, we can obtain the expression for $Q'_n(B_n)$, with which
the optimal $B^*_n$ can be solved using off-the-peg CVX tools. 
Then, the optimal power allocation $\mathbf{p}_n^*$ is obtained, as described in \textbf{Algorithm~\ref{Algorithm 2}}.

	\bibliographystyle{IEEEtran}
	\bibliography{ciations}

\begin{thebibliography}{10}
\providecommand{\url}[1]{#1}
\csname url@samestyle\endcsname
\providecommand{\newblock}{\relax}
\providecommand{\bibinfo}[2]{#2}
\providecommand{\BIBentrySTDinterwordspacing}{\spaceskip=0pt\relax}
\providecommand{\BIBentryALTinterwordstretchfactor}{4}
\providecommand{\BIBentryALTinterwordspacing}{\spaceskip=\fontdimen2\font plus
\BIBentryALTinterwordstretchfactor\fontdimen3\font minus
  \fontdimen4\font\relax}
\providecommand{\BIBforeignlanguage}[2]{{%
\expandafter\ifx\csname l@#1\endcsname\relax
\typeout{** WARNING: IEEEtran.bst: No hyphenation pattern has been}%
\typeout{** loaded for the language `#1'. Using the pattern for}%
\typeout{** the default language instead.}%
\else
\language=\csname l@#1\endcsname
\fi
#2}}
\providecommand{\BIBdecl}{\relax}
\BIBdecl

\bibitem{wahab_federated_2021}
O.~A. Wahab \emph{et~al.}, ``{Federated} machine learning: Survey, multi-level
  classification, desirable criteria and future directions in communication and
  networking systems,'' \emph{IEEE Commun. Surv. Tutor.}, vol.~23, no.~2, pp.
  1342--1397, 2nd Quart. 2021.

\bibitem{manias2021making}
D.~M. Manias and A.~Shami, ``Making a case for federated learning in the
  {i}nternet of vehicles and intelligent transportation systems,'' \emph{IEEE
  Netw.}, vol.~35, no.~3, pp. 88--94, May/June 2021.

\bibitem{Li2016Energy}
K.~Li \emph{et~al.}, ``Energy-efficient cooperative relaying for unmanned
  aerial vehicles,'' \emph{IEEE Trans. Mobile Comput.}, vol.~15, no.~6, pp.
  1377--1386, 2016.

\bibitem{9063670}
L.~U. Khan, I.~Yaqoob, N.~H. Tran \emph{et~al.}, ``Edge-computing-enabled smart
  cities: A comprehensive survey,'' \emph{IEEE Internet Things J.}, vol.~7,
  no.~10, pp. 10\,200--10\,232, Oct. 2020.

\bibitem{Li2019Energy}
K.~Li, R.~C. Voicu, S.~S. Kanhere, W.~Ni, and E.~Tovar, ``Energy efficient
  legitimate wireless surveillance of uav communications,'' \emph{IEEE Trans.
  Veh. Tech.}, vol.~68, no.~3, pp. 2283--2293, 2019.

\bibitem{9460016}
L.~U. Khan, W.~Saad, Z.~Han \emph{et~al.}, ``Federated learning for internet of
  things: Recent advances, taxonomy, and open challenges,'' \emph{IEEE Commun.
  Surv. Tutor.}, vol.~23, no.~3, pp. 1759--1799, 3rd Quart. 2021.

\bibitem{Wang2015VANET}
H.~Wang, R.~P. Liu, W.~Ni, W.~Chen, and I.~B. Collings, ``Vanet modeling and
  clustering design under practical traffic, channel and mobility conditions,''
  \emph{IEEE Trans. Commun.}, vol.~63, no.~3, pp. 870--881, 2015.

\bibitem{Li2019On}
K.~Li, W.~Ni, E.~Tovar, and A.~Jamalipour, ``On-board deep q-network for
  uav-assisted online power transfer and data collection,'' \emph{IEEE Trans.
  Veh. Tech.}, vol.~68, no.~12, pp. 12\,215--12\,226, 2019.

\bibitem{lim_federated_2020}
W.~Y.~B. Lim, N.~C. Luong, D.~T. Hoang \emph{et~al.}, ``Federated learning in
  mobile edge networks: A comprehensive survey,'' \emph{IEEE Commun. Surv.
  Tutor.}, vol.~22, no.~3, pp. 2031--2063, 3rd Quart. 2020.

\bibitem{yang_training_2020}
S.~Yang and Y.~Liu, ``Training efficiency of federated learning: A wireless
  communication perspective,'' in \emph{Proc. Int. Conf. WCSP}, Nanjing, China,
  Oct. 2020, pp. 922--926.

\bibitem{zeng_energy-efficient_2020}
Q.~Zeng, Y.~Du, K.~Huang \emph{et~al.}, ``Energy-efficient radio resource
  allocation for federated edge learning,'' in \emph{Proc. IEEE ICC Workshops},
  Dublin, Ireland, June 2020.

\bibitem{xu_client_2021}
J.~Xu and H.~Wang, ``Client selection and bandwidth allocation in wireless
  federated learning networks: A long-term perspective,'' \emph{IEEE Trans.
  Wireless Commun.}, vol.~20, no.~2, pp. 1188--1200, Feb. 2021.

\bibitem{7973146}
Z.~Ding, X.~Lei, G.~K. Karagiannidis \emph{et~al.}, ``A survey on
  non-orthogonal multiple access for 5{G} networks: Research challenges and
  future trends,'' \emph{IEEE J. Sel. Areas Commun.}, vol.~35, no.~10, pp.
  2181--2195, Oct. 2017.

\bibitem{abutuleb_joint_2020}
A.~Abutuleb, S.~Sorour, and H.~S. Hassanein, ``Joint task and resource
  allocation for mobile edge learning,'' in \emph{Proc. IEEE GLOBECOM}, Taipei,
  Taiwan, Dec. 2020, pp. 1--6.

\bibitem{pmlr-v54-mcmahan17a}
B.~McMahan, E.~Moore, D.~Ramage \emph{et~al.}, ``Communication-efficient
  learning of deep networks from decentralized data,'' in \emph{Proc. AISTATS},
  vol.~54, Apr. 2017, pp. 1273--1282.

\bibitem{9264742}
Z.~Yang, M.~Chen, W.~Saad \emph{et~al.}, ``Energy efficient federated learning
  over wireless communication networks,'' \emph{IEEE Trans. Wireless Commun.},
  vol.~20, no.~3, pp. 1935--1949, Mar. 2021.

\bibitem{fedasync}
C.~Xie, S.~Koyejo \emph{et~al.}, ``Asynchronous federated optimization,''
  \emph{arxiv:1903.03934}, 2019.

\bibitem{Fed-AT}
Z.~Chai, Y.~Chen, A.~Anwar \emph{et~al.}, ``Fed{A}{T}: A high-performance and
  communication-efficient federated learning system with asynchronous tiers,''
  in \emph{Proc. Int. Conf. High Perform. Comput. Netw. Storage Anal.}, NY,
  USA, Nov. 2021, pp. 1--17.

\bibitem{TT-fed}
X.~Zhou, Y.~Deng, H.~Xia \emph{et~al.}, ``Time-triggered federated learning
  over wireless networks,'' \emph{IEEE Trans. Wireless Commun.}, vol.~21,
  no.~12, pp. 11\,066--11\,079, Dec. 2022.

\bibitem{9725259}
Z.~Wang, Z.~Zhang, and J.~Wang, ``Asynchronous federated learning over wireless
  communication networks,'' \emph{IEEE Trans. Wireless Commun.}, vol.~21,
  no.~9, pp. 6961--6978, June 2022.

\bibitem{pmlr-v130-ruan21a}
Y.~Ruan, X.~Zhang, S.~Liang \emph{et~al.}, ``Towards flexible device
  participation in federated learning,'' in \emph{Proc. AISTATS}, Apr. 2021,
  pp. 3403--3411.

\bibitem{9322270}
X.~Ma, H.~Sun, and R.~Q. Hu, ``Scheduling policy and power allocation for
  federated learning in {NOMA} based {MEC},'' in \emph{Proc. IEEE GLOBECOM},
  Dec. 2020, pp. 1--7.

\bibitem{9718086}
Y.~Wu, Y.~Song, T.~Wang \emph{et~al.}, ``Non-orthogonal multiple access
  assisted federated learning via wireless power transfer: A cost-efficient
  approach,'' \emph{IEEE Trans. Commun.}, vol.~70, no.~4, pp. 2853--2869, Apr.
  2022.

\bibitem{8952884}
K.~Yang, T.~Jiang, Y.~Shi \emph{et~al.}, ``Federated learning via over-the-air
  computation,'' \emph{IEEE Trans. Wireless Commun.}, vol.~19, no.~3, pp.
  2022--2035, Mar. 2020.

\bibitem{9764370}
T.~Zhao, F.~Li, and L.~He, ``{DRL}-based joint resource allocation and device
  orchestration for hierarchical federated learning in {NOMA}-enabled
  industrial {IoT},'' \emph{IEEE Trans. Ind. Informatics.}, 2022 (Early
  Access).

\bibitem{9844152}
M.~Al-Abiad, M.~Hassan, and M.~Hossain, ``Energy efficient resource allocation
  for federated learning in {NOMA} enabled and relay-assisted internet of
  things networks,'' \emph{IEEE Internet Things J.}, vol.~9, no.~24, pp.
  24\,736--24\,753, Dec. 2022.

\bibitem{7557079}
M.~S. Ali, H.~Tabassum, and E.~Hossain, ``Dynamic user clustering and power
  allocation for uplink and downlink non-orthogonal multiple access ({NOMA})
  systems,'' \emph{IEEE Access}, vol.~4, pp. 6325--6343, Aug. 2016.

\bibitem{9484466}
W.~Ni, Y.~Liu, Z.~Yang \emph{et~al.}, ``Over-the-air federated learning and
  non-orthogonal multiple access unified by reconfigurable intelligent
  surface,'' in \emph{Proc. IEEE INFOCOM Workshops}, Vancouver, BC, Canada, May
  2021, pp. 1--6.

\bibitem{9829190}
J.~Zheng, H.~Tian, W.~Ni \emph{et~al.}, ``Balancing accuracy and integrity for
  reconfigurable intelligent surface-aided over-the-air federated learning,''
  \emph{IEEE Trans. Wireless Commun.}, vol.~21, no.~12, pp. 10\,964--10\,980,
  July 2022.

\bibitem{ota}
X.~Yu, B.~Xiao, W.~Ni \emph{et~al.}, ``Optimal power control for over-the-air
  federated edge learning using statistical channel knowledge,'' in \emph{Proc.
  WCSP}, 2022 (to appear).

\bibitem{SHI}
W.~Shi, Y.~Sun, X.~Huang \emph{et~al.}, ``Scheduling policies for federated
  learning in wireless networks: An overview,'' \emph{ZTE Commun.}, vol.~18,
  no.~2, pp. 11--19, June 2020.

\bibitem{7959539}
Y.~Gao, B.~Xia, K.~Xiao \emph{et~al.}, ``{Theoretical} analysis of the dynamic
  decode ordering {SIC} receiver for uplink {NOMA} systems,'' \emph{IEEE
  Commun. Lett.}, vol.~21, no.~10, pp. 2246--2249, June 2017.

\bibitem{boyd2004convex}
S.~Boyd and L.~Vandenberghe, \emph{Convex Optimization}.\hskip 1em plus 0.5em
  minus 0.4em\relax Cambridge University Press, 2004.

\bibitem{9796982}
L.~Cui, X.~Su, Y.~Zhou \emph{et~al.}, ``Optimal rate adaption in federated
  learning with compressed communications,'' in \emph{Proc. IEEE INFOCOM},
  London, United Kingdom, May 2022, pp. 1459--1468.

\bibitem{cvx}
M.~Grant and S.~Boyd, ``{CVX}: Matlab software for disciplined convex
  programming, version 2.1,'' Mar. 2014.

\bibitem{9387137}
C.~Sun, W.~Ni, and X.~Wang, ``Joint computation offloading and trajectory
  planning for {UAV}-assisted edge computing,'' \emph{IEEE Trans. Wireless
  Commun.}, vol.~20, no.~8, pp. 5343--5358, Mar. 2021.

\bibitem{chaoyanghe2020fedml}
C.~He, S.~Li \emph{et~al.}, ``Fed{ML}: A research library and benchmark for
  federated machine learning,'' \emph{arXiv:2007.13518}, 2020.

\bibitem{AdaptiveFL}
S.~Reddi, Z.~Charles, M.~Zaheer \emph{et~al.}, ``Adaptive federated
  optimization,'' \emph{arxiv:2003.00295}, 2020.

\bibitem{He_2016_CVPR}
K.~He, X.~Zhang, S.~Ren \emph{et~al.}, ``Deep residual learning for image
  recognition,'' in \emph{Proc. CVPR}, Las Vegas, US, June 2016.

\end{thebibliography}

\end{document}